%% file: ms.tex
\documentclass[acmtog, authorversion]{acmart}

\setcitestyle{square}

\usepackage[ruled]{algorithm2e} 

\SetAlFnt{\small}
\SetAlCapFnt{\small}
\SetAlCapNameFnt{\small}
\SetAlCapHSkip{0pt}
\IncMargin{-\parindent}

\acmJournal{TOG}
\acmVolume{VV}
\acmNumber{NN}
\acmArticle{XX}
\acmYear{YYYY}

\setcopyright{rightsretained}

\acmDOI{0000001.0000001_2}

\received{July 2017}

\input{99-pre}

\usepackage{graphicx}
\graphicspath{{figures/}}

\begin{document}

\title{Discrete Geodesic Nets for Modeling Developable Surfaces} 
\author{Michael Rabinovich}
\affiliation{%
  \institution{ETH Zurich}
  \department{Department of Computer Science}
  \streetaddress{Universit\"atstrasse 6}
  \city{Zurich}
  \postcode{8092}
  \country{Switzerland}}
\author{Tim Hoffmann}
\affiliation{%
  \institution{TU Munich}
  \department{Department of Mathematics}}
\author{Olga Sorkine-Hornung}
\affiliation{%
  \institution{ETH Zurich}
  \department{Department of Computer Science}
  \streetaddress{Universit\"atstrasse 6}
  \city{Zurich}
  \postcode{8092}
  \country{Switzerland}}

\renewcommand\shortauthors{Rabinovich, M. et al}

\begin{abstract}
\input{00-abstract}
\end{abstract}

%
%

%
%


\thanks{This work was supported in part by the ERC grant iModel (StG-2012-306877) and by the Deutsche Forschungsgemeinschaft-Collaborative Research Center, TRR 109, ``Discretization in Geometry and Dynamics.''

  Author's addresses: M.\ Rabinovich and O.\ Sorkine-Hornung: Department of Computer Science, ETH Zurich, Universit\"atstrasse 6, 8092 Zurich, Switzerland; Tim Hoffmann: Department of Mathematics, TU Munich, Boltzmannstrasse 3, 85748 Garching bei M\"unchen, Germany.}

\maketitle

\input{01-introduction}

\input{02-preliminaries}
\input{03-related-work}
\input{04-notations}

\input{05-orthogonal-geodesic-nets}

\input{05-p-modeling-dev-surfaces}
\input{05-q-parallels-to-smooth}
\input{06-isometry-4Q}

\input{07-conclusion}

\begin{acks}

The authors would like to thank Noam Aigerman, Mario Botsch, Oliver Glauser, Roi Poranne, Katja Wolff, Christian Sch\"uller, Jan Wezel and Hantao Zhao for illuminating discussions and help with results production.

The work was supported in part by the \grantsponsor{ERC}{European Research Council}{https://erc.europa.eu/} under Grant
No.:~\grantnum{ERC}{StG-2012-306877} (ERC Starting Grant iModel) and by the Deutsche Forschungsgemeinschaft-Collaborative Research Center, TRR 109, "Discretization in Geometry and Dynamics."  
\end{acks}

\bibliographystyle{ACM-Reference-Format}
\bibliography{97-bib-opt}

\appendix
\input{96-appendix}
\end{document}

%% file: 99-pre.tex
\usepackage{amsmath}
\usepackage{amssymb}
\usepackage{amsfonts}
\usepackage{enumitem}

\usepackage{microtype}
\usepackage{units}

\newenvironment{tight_enumerate}{
\begin{enumerate}[topsep=4pt,itemindent=\parindent]
  \setlength{\itemsep}{1pt}
  \setlength{\parskip}{0pt}
  \setlength{\parsep}{0pt}
}{\end{enumerate}}

\usepackage{booktabs}

\usepackage{dsfont}
\usepackage{wrapfig}

\usepackage{mdframed}
\mdfsetup{skipabove=\topskip,skipbelow=\topskip}
\newmdenv[	linewidth=0.75pt,
			leftline=false,
			rightline=false,
   			innerleftmargin=0pt,
			innerrightmargin=0pt]{topbot}

\newtheorem{mydefinition}{Definition}

\newcommand{\figref}[1]{Fig.\ \ref{#1}}
\newcommand{\secref}[1]{Sec.\ \ref{#1}}
\newcommand{\equref}[1]{Eq.\ \eqref{#1}}
\newcommand{\defref}[1]{Def.\ \ref{#1}}
\newcommand{\lemmaref}[1]{Lemma \ref{#1}}
\newcommand{\theoremref}[1]{Theorem \ref{#1}}
\newcommand{\corollaryref}[1]{Cor.\ \ref{#1}}

\DeclareMathOperator*{\argmin}{arg\,min}

\newcommand{\R}{\mathds{R}}
\newcommand{\Z}{\mathds{Z}}

\newcommand{\x}{x}
\newcommand{\p}{p}

\newcommand{\Fbo}{F_{\bar{1}}}
\newcommand{\Fbt}{F_{\bar{2}}}

\newcommand{\norm}[1]{\left\Vert#1\right\Vert}

%% file: 00-abstract.tex

We present a discrete theory for modeling developable surfaces as quadrilateral meshes satisfying simple angle constraints. The basis of our model is a lesser known characterization of developable surfaces as manifolds that can be parameterized through orthogonal geodesics. Our model is simple, local, and, unlike previous works, it does not directly encode the surface rulings. This allows us to model continuous deformations of discrete developable surfaces independently of their decomposition into torsal and planar patches or the surface topology. We prove and experimentally demonstrate strong ties to smooth developable surfaces, including a theorem stating that every sampling of the smooth counterpart satisfies our constraints up to second order. We further present an extension of our model that enables a local definition of discrete isometry.
We demonstrate the effectiveness of our discrete model in a developable surface editing system, as well as  computation of an isometric interpolation between isometric discrete developable shapes.

%% file: 01-introduction.tex

\section{Introduction}

\begin{figure}
\includegraphics[width=\linewidth]{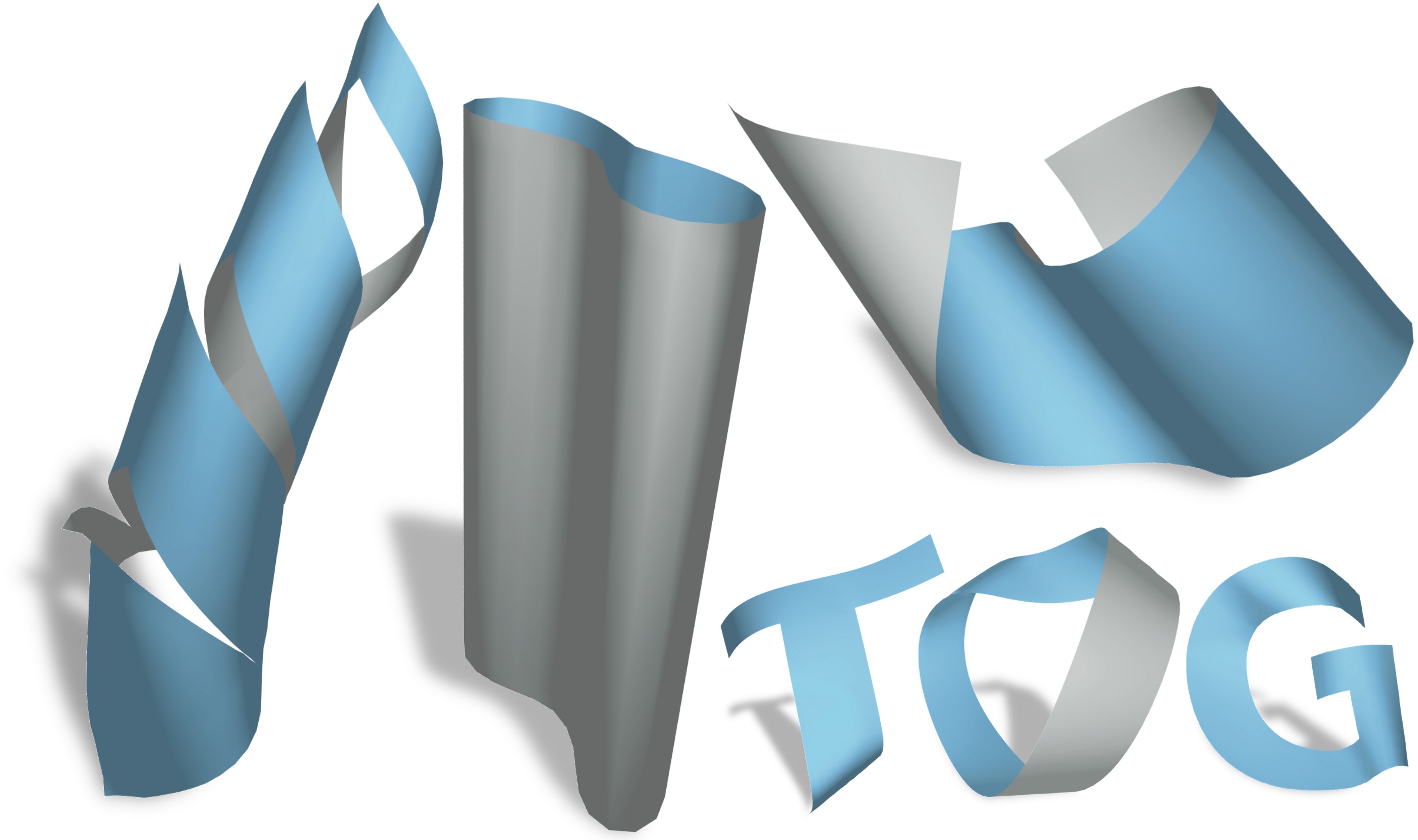}
\caption{\label{fig:teaser}We propose a discrete model for developable surfaces. The strength of our model is its locality, offering a simple and consistent way to realize deformations of various developable surfaces without being limited by the topology of the surface or its decomposition into torsal developable patches.  Our editing system allows for operations such as freeform handle-based editing, cutting and gluing, modeling closed and un-oriented surfaces, and seamlessly transitioning between planar, cylindrical, conical and tangent developable patches, all in a unified manner.}	
\end{figure}

\begin{figure*}
\centering
\includegraphics[height=0.21\linewidth]{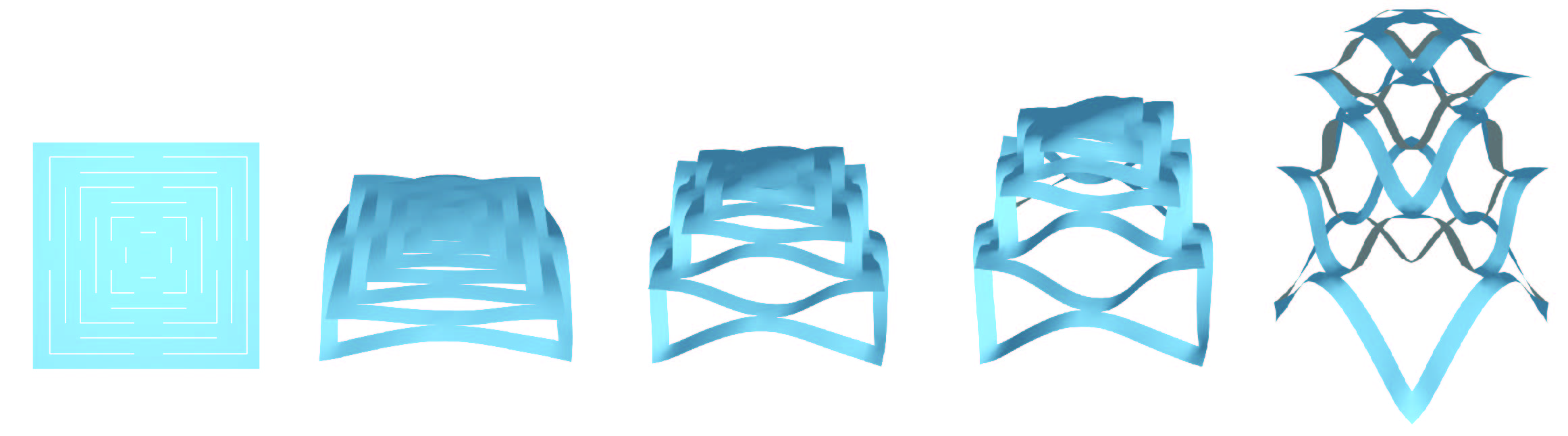}
\includegraphics[height=0.24\linewidth]{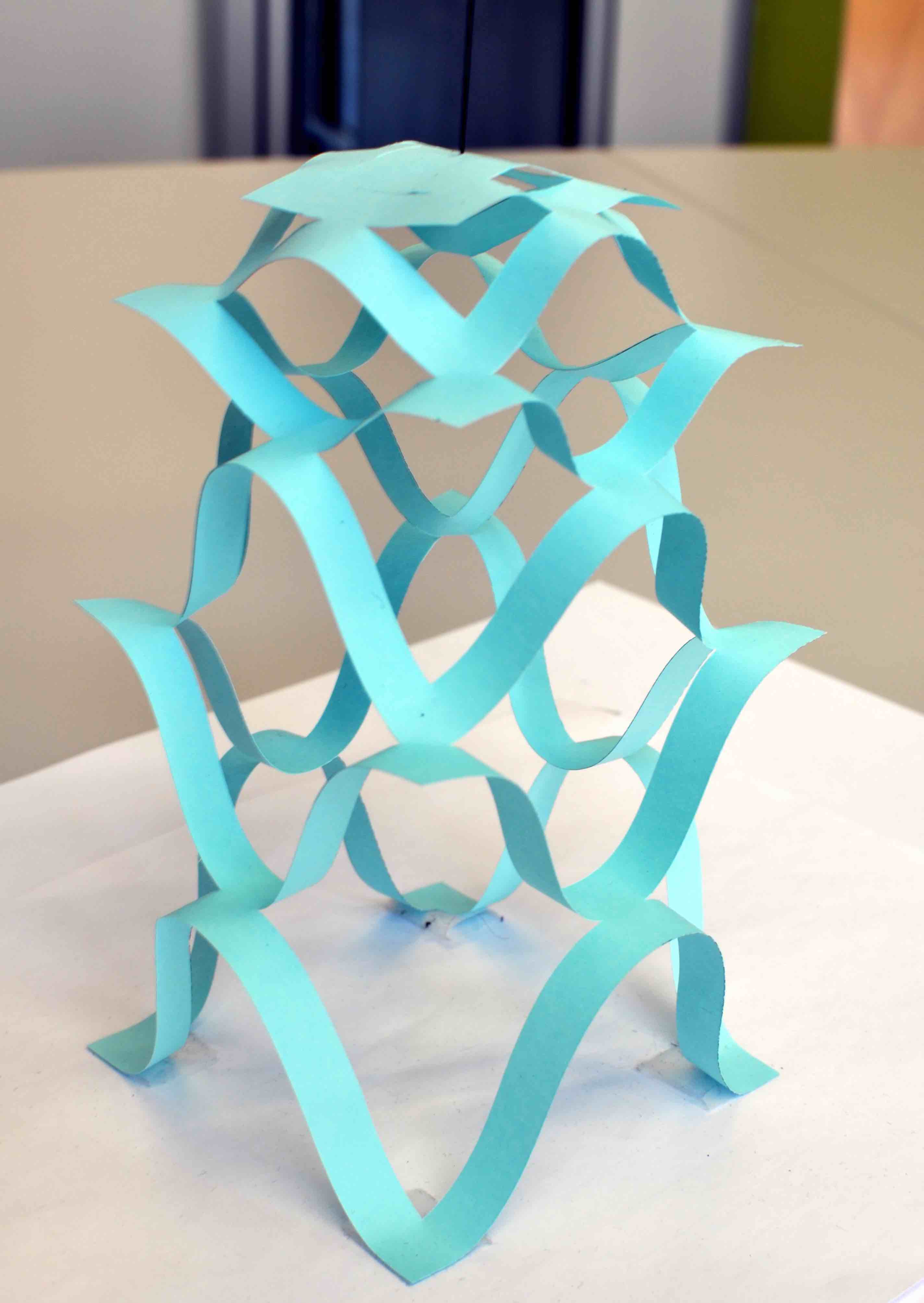}
\caption{\label{fig:lantern}
A lantern shaped developable surface, demonstrating how our discrete model can seamlessly and effortlessly model developable surfaces with nontrivial topology. This figure was created from an alternating cut pattern on a square sheet (left). The shapes in the middle were formed by pulling the central vertex up while constraining the corners to stay on their initial plane.  Right: A physical model made of paper with the same cut pattern (we glued the corners to the table and lifted the center point using a thin thread).
}	
\end{figure*}

The concept of isometry lies at the core of the study of surfaces. Loosely speaking, two surfaces are isometric if one can be obtained by bending and twisting the other. The deforming map is then called an isometry, and the properties of a surface that are invariant to isometries are called intrinsic properties. A \emph{local} isometry is such a mapping in a neighborhood of some point on a surface.


Surfaces that are locally isometric to a plane are called developable surfaces. In the physical world, these surfaces can be formed by bending thin flat sheets of material, which makes them particularly attractive in manufacturing \cite{ships}, architecture \cite{gehry} and art \cite{huffman2}. Consequently, the  design of freeform developable surfaces has been an active research topic in computer graphics, computer aided design and computational origami for several decades. 

{The scope of our work} is modeling developable surfaces through \emph{deformation}, which can be applied in a design and fabrication pipeline. This is in contrast to contour interpolation works \cite{Frey2,sheffer}, which compute a developable surface passing through an input set of curves, as well as shape approximation through developable surfaces \cite{pottmann_approx,chen_approx}. Our goal is to model \emph{smooth} deformations, such as the rolling and bending of a planar sheet into a cone, rather than $C^0$ origami-like folding and creasing \cite{tachi_sim}. 

{Smooth developable surfaces} are 	well studied in differential geometry \cite{do_carmo} and are often characterized as surfaces with vanishing Gaussian curvature, or, equivalently, as ruled surfaces with a constant normal along each ruling.


A given smooth developable surface $S$ can be naturally discretized as a ruled surface, as it can be locally represented by a single curve 
\setlength{\intextsep}{8pt}%
\setlength{\columnsep}{8pt}%
\begin{wrapfigure}{r}{0.3\linewidth}
  \centering
  \includegraphics[width=\linewidth]{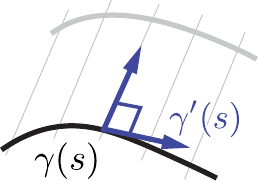}
\end{wrapfigure}
and its orthogonal rulings (see inset). For this reason, many discrete developable models encode rulings explicitly \cite{rect_dev,conical}. However, this representation has limitations when it comes to interactive modeling of a developable surface. 
In this process, the user starts with an initial developable surface $S_0$, for instance a planar surface, and interactively manipulates it to obtain a desired surface $S$ (see Figs.\ \ref{fig:teaser}, \ref{fig:lantern}). Since the output surface is not necessarily known precisely in advance, one would like to explore the entire space of attainable developable surfaces in this interactive setting. As stated in \cite{pottmann_new}, explicitly including the rulings in the surface representation limits the space of possible deformations of $S_0$. From a user's point of view, it may be more intuitive to manipulate local point handles to edit the surface, rather than editing its global rulings. 

 We show that such developable shape space exploration is made possible by discretizing a lesser known, local condition for developability: The existence of an orthogonal geodesic parameterization.
We propose an alternative way to understand developable surface \emph{isometries} by looking at their invariants, rather than the rulings. 

\subsection{Contributions}
\begin{itemize}[label={--}]
  \item We introduce \textit{discrete orthogonal geodesic nets} to model developable surfaces as quadrilateral nets with angle constraints. Our conditions are simple and local, and our model does not depend on the explicit encoding of the rulings or the surface topology.
  \item We use this model to build a simple editing system for developable surfaces with point handles as user interface. Our system can smoothly transition between a wide range of shapes while maintaining developability, and, unlike previous methods, does not require the user to specify global rulings or any other global structure of the unknown desired shape.
  \item We further study our new discrete model and draw parallels to smooth developable surfaces. We prove that our discrete constraints are satisfied in the smooth case up to second order, analyze our model's degrees of freedom, discretize quantities such as tangents and normals and propose a local scheme to approximate the rulings. We formulate and prove a discrete analogue to a known continuous theorem linking curvature line parameterizations, geodesic parameterizations, and developable surfaces.
  \item We introduce a generalization of our nets, called  \emph{discrete 4Q orthogonal geodesic nets}, which allows us to define {local discrete isometry} between our surfaces. We demonstrate the effectiveness and flexibility of such 4Q nets by computing an {isometric interpolation} between isometric developable shapes.
\end{itemize}



%% file: 02-preliminaries.tex

\section{Preliminaries}
\label{sec:preliminaries}

\subsection{Nets in discrete differential geometry}
In the spirit of previous works in discrete differential geometry \cite{DDG_Desbrun,DDG_Book}, we discretize a developable surface as a quad grid mesh, referred to as a \emph{net}, which can be viewed as a discrete analogue to a smooth parameterization (often termed \emph{smooth net}). This approach has been previously taken to discretize and construct a variety of surface types, including constant Gaussian curvature surfaces \cite{wunderlich,ddg_neg_gaussian}, minimal surfaces \cite{ddg_minimal} and isothermic surfaces \cite{ddg_isothermic}. Just as a smooth surface can be locally represented by a parameterization $f:\R^2 \rightarrow \R^3$, a discrete surface can be locally represented by a discrete map $F:\Z^2 \rightarrow \R^3$ (\figref{fig:discrete_chart}). This structural view is especially appealing, as it can be used to convert between smooth and discrete notions on surfaces, such as tangents, normals and surface transformations, and to analyze the construction of discrete surfaces and their convergence to the continuous counterparts. Discrete analogues of smooth differential geometry theorems are systematically studied in the context of nets; see the review in \cite{DDG_Book}.

\begin{figure}[h]
\centering
   \includegraphics[width=\linewidth]{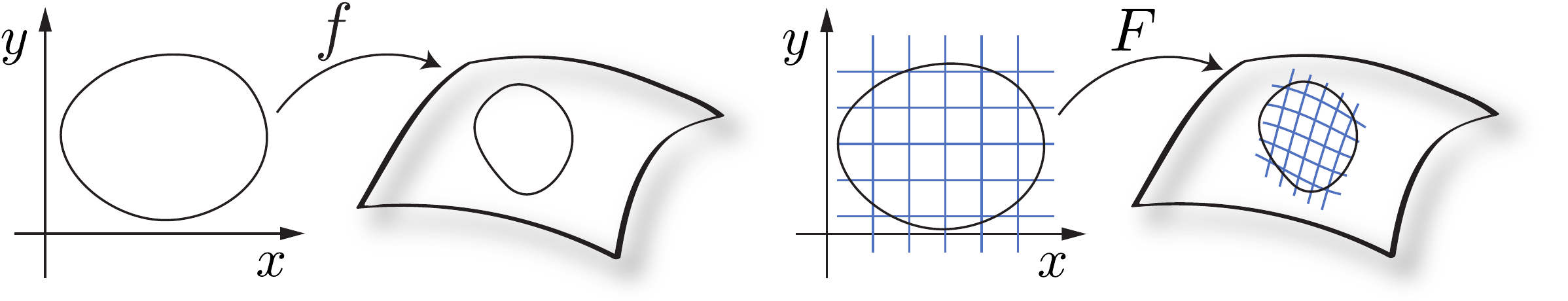}
   \caption{\label{fig:discrete_chart} A smooth net $f:\R^2 \rightarrow \R^3$ and a discrete net $F:\Z^2 \rightarrow \R^3$.}
\end{figure}

The same smooth surface can be represented by many different parameterizations, or nets, and some are more convenient than others.  These typically differ by the properties of their coordinate curves $f(x_0+t,y_0),f(x_0,y_0+t)$.
Prominent examples include curvature line nets, where the coordinate curves are principal curvature lines, and asymptotic nets, whose coordinate curves trace the asymptotic directions of a surface. The freedom to choose various nets exists also in the discrete setting, and usually a discrete model of a surface is coupled with a given parameterization. For example,  discrete minimal surfaces have been defined through curvature line nets, and discrete constant negative Gaussian curvature surfaces through nets of asymptotic lines \cite{ddg_neg_gaussian,ddg_minimal}. Each choice of parameterization implies certain conditions on the discrete surface, formulated in terms of the values of $F$, i.e., the positions of the net's vertices. 

\subsection{Developable surfaces through conjugate nets}
\label{sec:dev_curvature}
The neighborhood of a non-planar point $\p$ on a developable surface $S$ can be locally parameterized by its rulings, which are straight lines contained in the surface. This means that there exists a neighborhood $U \subseteq S$ such that $\p \in U$ and all points in $U$ are parameterized by
$$\x(s, t) = \gamma(s) + t\, {r}(s),$$ where ${r}(s)$ corresponds to the ruling, and fixing the parameter $t$ gives us another curve on the surface with non-vanishing curvature, from which the rulings emanate.

The subset $U \subseteq S$ is called a \emph{torsal surface}. A torsal surface can be classified based on the directions of its rulings: if they are parallel, it is said to be cylindrical, if they all intersect at a single point, it is a generalized cone, and otherwise it is a so-called tangent surface (see \figref{fig:torsal_surfaces}).
\begin{figure}[t]
   \includegraphics[width=\linewidth]{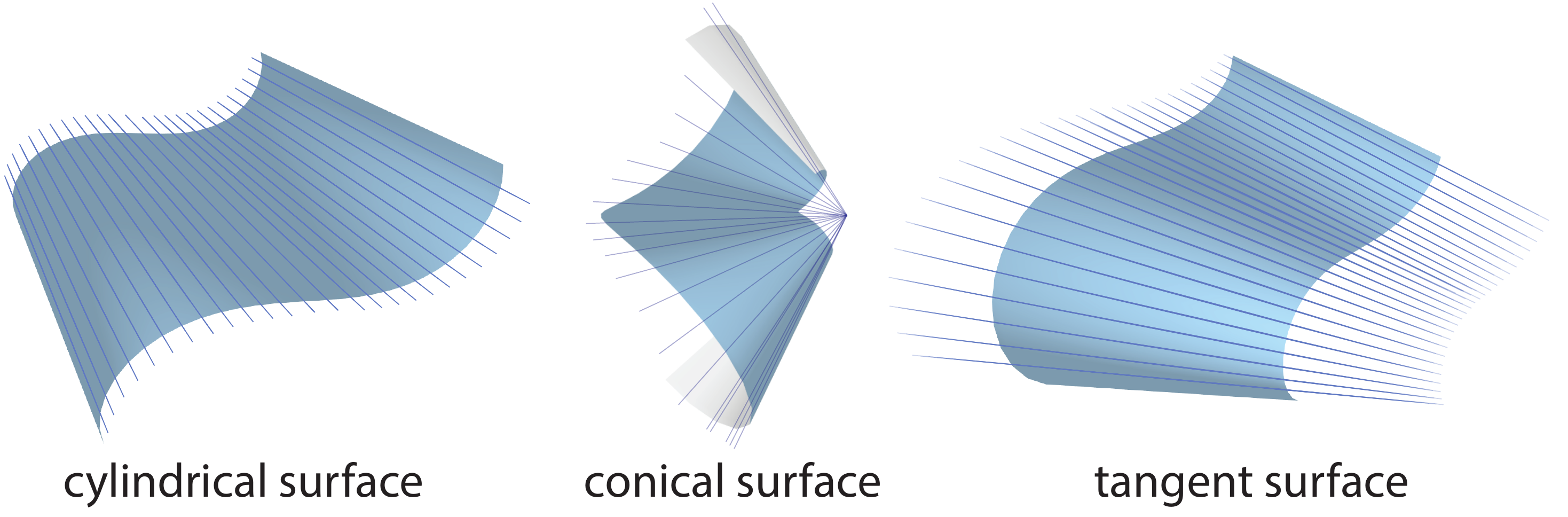}
   \caption{  \label{fig:torsal_surfaces} Different types of developable surfaces and their rulings. Left: A cylindrical shape with parallel rulings. Center: all rulings meet in a point in a conical surface. Right: a tangent developable surface where the rulings meet at a curve.}
\end{figure}


A parameterization of a developable surface through its rulings is called a developable \emph{conjugate net} \cite{conical}.
To clarify the previous statement, we elaborate on the definition of a conjugate net in a more general context, where $f$ is a smooth net that is not necessarily developable. A smooth parameterization $f$ is a conjugate net if 
$$\langle {n}_x, f_y \rangle = 0, \ \ \textrm{where }{n} = \frac{f_x \times f_y}{\norm{f_x \times f_y}}.$$
Here, $n$ is the normal map of $f$, and a subscript denotes differentiation with respect to the coordinate in the subscript. In this case the tangents $f_x, f_y$ are said to be conjugate directions. The condition is equivalent to $f_{xy} \in \mathrm{span}\{f_x, f_y\}$. Intuitively, in such a parameterization, infinitesimally small squares in the parameter domain are mapped to planar quads on the surface up to second order. Hence, planar quad meshes are seen as a discretization of conjugate nets \cite{DDG_Book}. Note that curvature line nets are a special case of conjugate nets. In the case of a developable surface, the normal $n$ is constant along a ruling, and therefore any developable net parameterized through rulings is in fact a conjugate net. A well established \emph{discrete} model for a conjugate developable net is a planar quad strip \cite{sauer,computational_line,conical}.

\subsection{The combinatorics of a developable surface}
\label{sec:combinatorial_problem}

The possible presence of  planar parts in developable surfaces further complicates their representation. A general developable surface is a composition of (possibly infinite) torsal and planar patches. The works of Liu et al.~\shortcite{conical} and Kilian et al.~\shortcite{curved_folding} model torsal surfaces by discrete conjugate nets, i.e., planar quad strips where the rulings are explicitly given by the transversal quad edges.
Accordingly, the discrete representation in \cite{conical,curved_folding,pottmann_new} consists of multiple discrete torsal patches connected together to form a discrete developable surface. The connectivity between those patches is represented by a combinatorial structure termed \textit{decomposition combinatorics} \cite{pottmann_new}. As stated in \cite{pottmann_new}, this fixed combinatorial structure requires the user to manually specify the said combinatorial structure of the modeled developable surface, and it is not possible to model a smooth transition between different combinatorial decompositions. We call this problem the \emph{combinatorial problem} (see \figref{fig:combinatoric_problem}). 

\begin{figure}[t]
\centering
   \includegraphics[width=\linewidth]{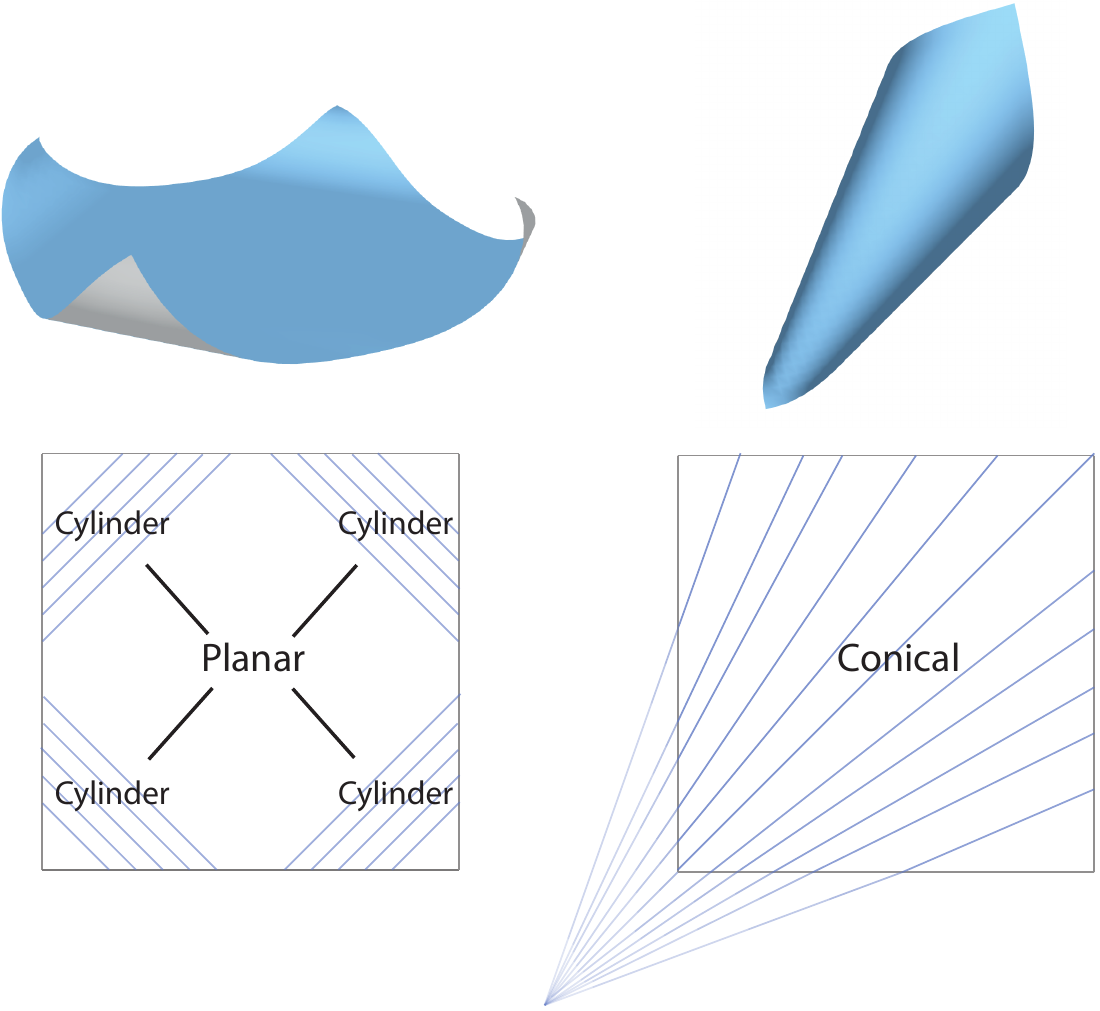}
   \caption{Two isometric shapes (top row) that are composed of different configurations of torsal surfaces, as illustrated in the bottom row. Our method does not rely on explicit encoding of this combinatorial structure of the developable surface, and can seamlessly model the transition between shapes without additional input from the user.}
   \label{fig:combinatoric_problem}
\end{figure}

\subsection{Developable surface isometry} \label{sec:dev_surface_iso}
A common task in a developable surface editing system is modeling isometries, which are non-stretching deformations that preserve distances on the surface. 
Since we are interested in modeling and editing shapes while staying within the shape space of developable surfaces, surface representations by conjugate or curvature lines are not good candidates for this application, because they are not invariant under isometries, as we explain next.


Isometrically unrolling a developable surface onto the plane reveals the innate shape of its curves. In the following, we often display a developable surface next to its flattened, isometric planar version, and refer to the geometry of a thereby flattened curve as the curve's intrinsic geometry. For instance, geodesics on a developable surface are curves that are intrinsically straight. 
The intrinsic shape of a curve is determined by its geodesic curvature $\kappa_g$, which does not change under isometry. 
As an example, all curvature lines of a cylinder are intrinsically straight, and for a cone they are a family of concentric circles and radial straight lines emanating from a single point (see \figref{fig:curve_flat}). Rulings and their conjugate directions are altered by isometries, but an isometry always maps geodesics to geodesics and intrinsic circles to circles of the same geodesic curvature. Therefore, a discrete isometry cannot be plausibly defined based on a mapping between conjugate curves on two developable surfaces.

\begin{figure}[t]
\centering
   \includegraphics[width=\linewidth]{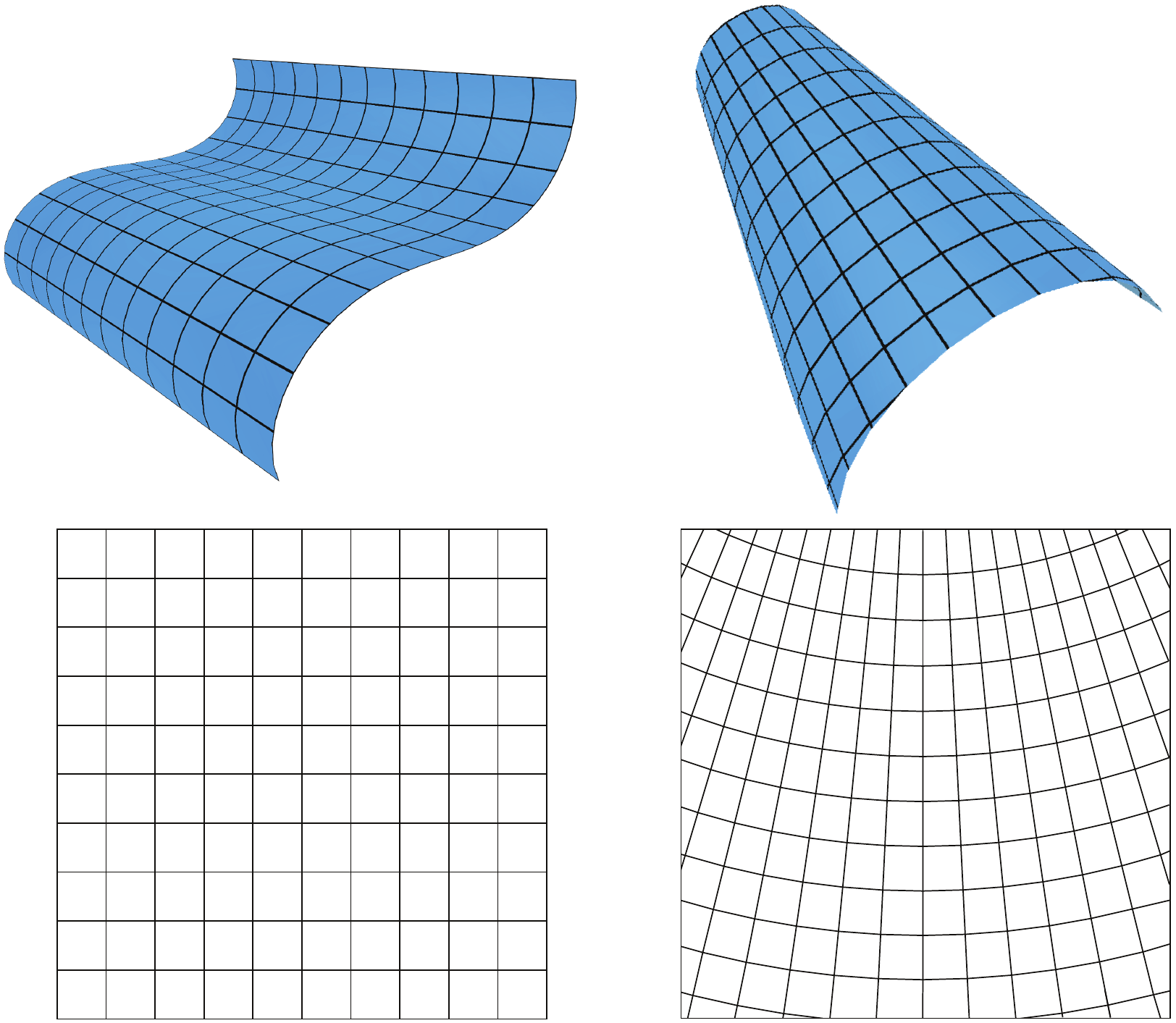}
   \caption{\label{fig:curve_flat}Two isometric developable surfaces (a cylindrical and a conical one) in curvature line parameterization (top row), and their isometrically flattened versions (bottom row), which reveal the intrinsic geometry of the curvature lines. These families of lines are intrinsically different, and there is no isometric mapping from one family to the other.}
\end{figure}

\subsection{Developable surface through orthogonal geodesic nets}
We propose to look at a different type of parameterization of developable surfaces, which is better suited for our interactive editing goals and is a more natural starting point to define discrete isometry.
Imagine taking a flat piece of paper with a square grid texture and watching the vertices of the squares while curving and rolling the paper. Squares, which started as planar, transform, but the intrinsic distances between all points stay the same, as long as one does not tear or stretch the paper. This is analog to our model. We propose a discrete model of developable surfaces through intrinsic entities: geodesics, which are invariant under isometries. A net $f$ is a geodesic net if its coordinate curves trace geodesics on the surface. On a developable surface, geodesics are straight lines when developed onto the plane. As we still have a degree of freedom in choosing the directions of the intrinsic lines, we set them to be of the simplest form -- orthogonal -- as in a rectangular grid (see \figref{fig:geo_net_prelimanary}).
By employing geodesics rather than rulings and conjugate directions, we overcome the aforementioned combinatorial problem and are able to define a notion of discrete isometry for such surfaces.

\begin{figure}[t]
\centering
   \includegraphics[width=\linewidth]{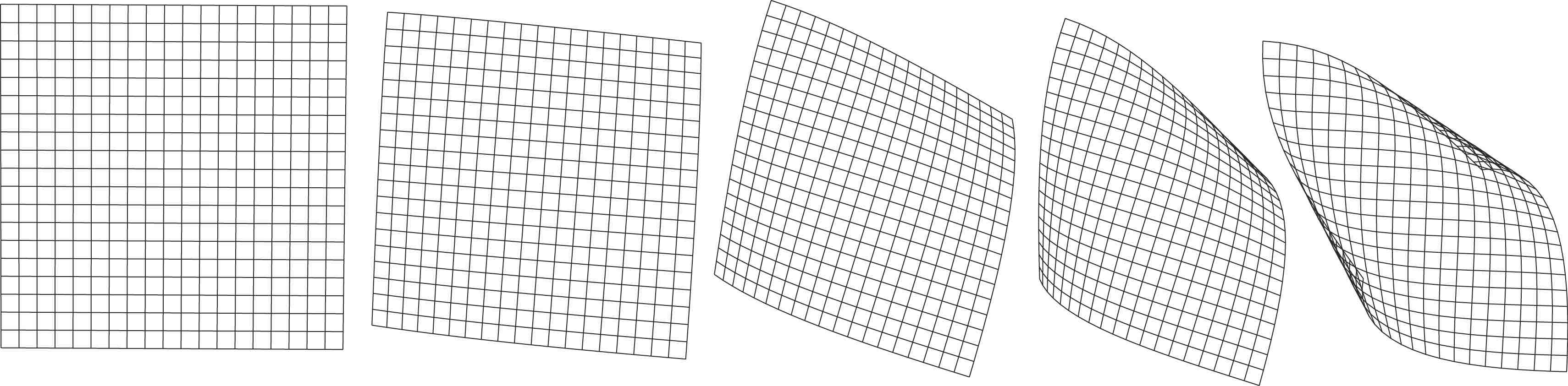}
   \caption{\label{fig:geo_net_prelimanary} Tracing orthogonal geodesics while rolling a planar surface into a circular cone.}
\end{figure}


%% file: 03-related-work.tex

\section{Related work}
\subsection{Developable surfaces}
The theory of surfaces formed by local $C^2$ isometries of the plane is covered in the differential geometry literature \cite{do_carmo,spivak} and traces back to the works of Euler and Monge in the eighteenth century \cite{DevHistory}. Gauss' Theorema Egregium coupled with Minding's theorem shows that $C^2$ developable surfaces are surfaces with zero Gaussian curvature. Intuitively, this means that the image of their Gauss map is a curve or a point. Another point of view is the characterization of developable surfaces as special ruled surfaces, namely, those with constant tangents along rulings \cite{computational_line}. Hence, a developable surface is locally a planar or a torsal surface. A torsal surface can be constructed by a single curve: For example, one can pass a torsal surface through a curvature line curve and its parallel Bishop frame \cite{bishop}, or through a geodesic and its Frenet frame \cite{graustein}.

The study of $C^1$ and $C^0$ developable surfaces is a much newer area, stirred by the beautiful models and work of Huffmann \cite{huffman,huffman2} and more recently by the field of computational origami \cite{origami_book}, which examines shapes created by straight and curved folds. Straight folds are $C^0$ creases through lines on a paper. Any shape created by repeated application of these folds is piecewise planar \cite{non_pleated}. Curved folds are $C^0$ creases through arbitrary curves on a paper. These are more rigid than straight folds, as splitting a surface into two parts by a curve and folding the surface on one side of the curve locally determines the shape of the other part \cite{curved_folding}.

The study of smooth developable surfaces is analytic in nature, whereas the study of origami folds is in essence combinatorial.
As previously stated, our work focuses on modeling smooth deformations.



\subsection{Modeling with developable surfaces}
Works on deformations of developable surfaces can be largely categorized into \emph{geometric} and \emph{physics based}.

{Geometric approaches} are not tied to a physical representation such as a paper sheet or a metal plate. They mainly consider and discretize  the geometry of a smooth developable surface. The foundation of these works is a discrete developable surface model, i.e., an \emph{exact} definition of the set of discrete developable surfaces. The definition should be flexible enough for the user to explore a wide range of shapes, while capturing important properties of the smooth surface.
The works of Liu et al.~\shortcite{conical} and Kilian et al.~\shortcite{curved_folding} model torsal surfaces as planar quad strips, which are a discretization of  developable conjugate nets. In \cite{pottmann_new} the authors model smooth torsal surfaces as developable splines. These are represented as ruled surfaces connecting two B\'{e}zier curves satisfying a set of quadratic equations that guarantee a constant normal along rulings. The work of \cite{rect_dev} models a torsal surface by a single geodesic curve and rulings emanating from it, i.e., the \emph{rectifying developable}  of a curve. All works above model a general developable surface as a composition of multiple torsal surfaces, explicitly encoding rulings and sharing the combinatorial problem we discussed in \secref{sec:combinatorial_problem}. Moreover, by construction these approaches cannot model isometry between different torsal shapes, such as a cylinder and a cone, as explained in \figref{fig:curve_flat}. We refer the reader to the 'Limitations' and 'Future work' paragraphs in Section~7 of \cite{pottmann_new} for an in-depth discussion of these shortcomings. 

The work of \cite{solomon} presents an origami based editing system for developable surfaces, allowing the user to navigate through the highly nonlinear space of admissible folds of a sheet. By involving a mean curvature bending energy, the user can further ask to relax the folds, resulting in a smoother looking, yet always piecewise planar surface~\cite{non_pleated}. Due to the reliance on global folds, this method shares a similar dependency on rulings with the previously mentioned works, which also complicates the user interface. 
Our proposal can be seen as a follow-up to all these works, removing the dependency on rulings and adding a notion of discrete isometry that is capable of smoothly interpolating between a wide range of shapes.

In contrast to geometric models, {physics based models} are coupled with given material properties of the surface. They model a material's behavior through energy minimization, simulating the physical shape when applying forces. The focus of these works is the physics of an object, such as an elastic simulation \cite{shells}, or paper crumpling and tearing \cite{Narain,Wang,SchreckEG2017}, rather than the geometry of developable surfaces. As such, these works do not deal with defining a precise notion of a discrete developable surface, nor do they aim at the exploration of the entire shape space of developable surfaces without straying off constraints. Developable surface editing can be indirectly approximated by discrete shell models \cite{grin_shells,froh_botsch} when starting from a flat sheet and setting a very high penalty in the stretch component of the elastic energy; the latter could lead to numerical problems, however. We view the physics based approaches as tangential to the geometric models, and they can also potentially benefit from new discrete surface models.

\subsection{Developable surfaces in discrete differential geometry}
As mentioned in \secref{sec:preliminaries}, the work of Liu et al.~\shortcite{conical} discretizes developable surfaces through conjugate line nets as planar quad strips, where the transversal quad edges lie on rulings. In contrast, our proposed discretization is through orthogonal geodesic nets, which is especially convenient when modeling deformations and isometries of developable surfaces. Our discretization is inspired by the work of Wunderlich~\shortcite{wunderlich} on discrete Voss surfaces, which are surfaces parameterized through conjugate lines that are also geodesics. Voss surfaces include surfaces that are not necessarily developable, and  
modeling with conjugate \emph{orthogonal} geodesics is quite limiting, since any such net is in fact a cylindrical shape.
Therefore, as a base for our model we use the same notion of a geodesic net set by Wunderlich but drop the conjugacy requirement, which means that our model allows for non planar quads. 

A few works in DDG cover discrete isometries of specific classes of surfaces, such as those of Voss surfaces~\cite{Schief2008}, where conjugate geodesics are preserved. We are not aware of a method that covers the entire range of developable surface isometries. As mentioned, developable Voss surfaces form only a limited subset of developable surfaces, and our isometry definition subsumes this subset, covering general developable surfaces in orthogonal geodesic parameterization. 

%% file: 04-notations.tex
\section{Notations and setup}

\begin{figure}[b]
\centering
   \includegraphics[width=\linewidth]{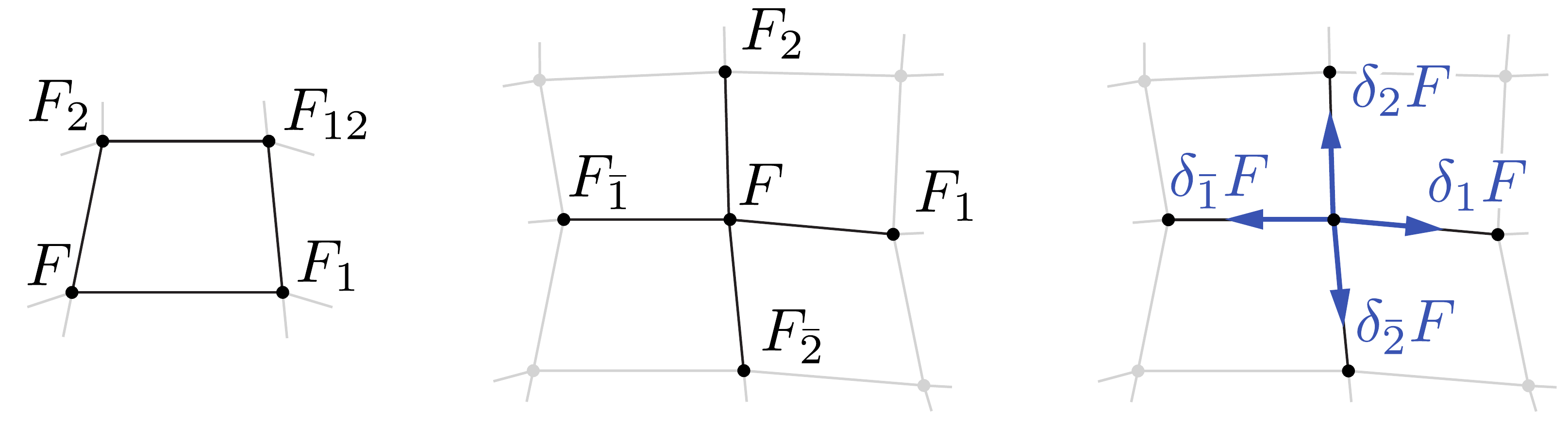}
   \caption{The shift notation on a quad (left) and a star (center); edge directions (right).}
   \label{fig:shift_notation}
\end{figure}

As briefly introduced in \secref{sec:preliminaries}, we denote continuous maps in lower case letters and their discrete equivalents by upper case. 
The notation $f(x,y): U \rightarrow \R^3$, where $U \subseteq \R^2$, refers to a (local) regular parameterization of a smooth surface, and $n(x,y):U \rightarrow \mathcal{S}^2$ is its normal map. Derivatives with respect to the coordinates $x$ and $y$ are denoted by subscripts, e.g., tangent vectors  $f_x, f_y$ and derivatives of the normal $n_x, n_y$. We denote the unit tangents of the coordinate lines by $t_1 = f_x/\|f_x\|, \ t_2 = f_y/\|f_y|$, which are linearly independent as $f$ is an immersion.

A natural discrete analogy for a local parameterization $f$ is a map $F: V \rightarrow \R^3$, where $V \subseteq \Z^2 $. We refer to $F$ as our discrete net, and likewise $N: V \rightarrow \mathcal{S}^2$ denotes our discrete Gauss map. Discrete unit tangents are denoted by $T_1, T_2$. We define these quantities in the following for our particular setting, namely discrete geodesic nets.

As is customary in discrete differential geometry, we slightly abuse the naming and employ shift notation to refer to vertex positions on our net, denoting 
\begin{align*}
&F = F(j,k),\ F_1 = F(j+1,k),\ F_2 = F(j,k+1),\\ 
&F_{12} = F(j+1,k+1),\ \Fbo = F(j-1,k),\ \Fbt = F(j, k-1),
\end{align*}
where $j, k \in \Z$, i.e., the lower index denotes the coordinate number to shift, and a bar above it indicates a negative shift  (see \figref{fig:shift_notation}). The unit-length directions of edges emanating from a point $F$ are denoted as $\delta_{1}F,\  \delta_{2}F,\ \delta_{\bar 1}F,\ \delta_{\bar 2}F,$ i.e., 
\begin{align*}
\delta_{1}F = (F_1 - F)/\|F_1-F\|,\ \ \ \delta_{\bar 1}F = (\Fbo - F)/\|\Fbo-F\|,\\
\delta_{2}F = (F_2 - F)/\|F_2-F\|,\ \ \ \delta_{\bar 2}F = (\Fbt - F)/\|\Fbt-F\|.
\end{align*}
%
We assume our net is a discrete immersion, which means that the edge directions $\delta_iF,\delta_{\bar i}F$ are distinct.
In practice, we represent our discrete nets as pure quad grid meshes, where the valence of every inner vertex is 4. We refer to an inner vertex, its four neighbors and its four emanating edges as a \emph{star}. 
Our discrete nets neither require nor assume any global orientation on the mesh. The shift notation requires only a local arbitrary orientation per quad or star, and is used for convenience.

%% file: 05-orthogonal-geodesic-nets.tex

\section{Discrete orthogonal geodesic nets} \label{sec:disc_orth_nets}


We are interested in defining conditions on $F$, i.e., on the positions of our mesh vertices, such that it represents a discrete developable surface parameterized by orthogonal geodesic lines. In the following, we develop the necessary definitions and their properties, to arrive at the following condition:

\begin{topbot}
\begin{mydefinition}
\label{def:def1}
A discrete net $F$ is said to be a discrete developable surface in orthogonal geodesic parameterization, i.e., a discrete orthogonal geodesic net, if for every star, all angles between consecutive edges are equal. 
\end{mydefinition}
\end{topbot}

To develop the rationale for the condition above, we start by looking at smooth developable geodesic nets.
\subsection{Smooth developable geodesic nets}
When is a geodesic net a developable net? Let $f:\R^2 \rightarrow \R^3$ be a geodesic net and $P = \{(x,y)\in \R^2\ |\ x_0 \leq x \leq x_1,\ y_0 \leq y \leq y_1 \}$ an axis-aligned rectangle. The rectangle is mapped by $f$ to a ``curved rectangle'' $f(P)$. Let $\alpha_j,\  j = 1, \ldots, 4,$ be the interior angles at the vertices of $f(P)$, measured as the angles between the respective tangent directions (as usual in differential geometry), e.g., $\alpha_1 = \sphericalangle\left(f_x(x_0,y_0),\ f_y(x_0,y_0)\right)$.

\begin{lemma}
\label{Lem:geo_dev_cond}
A geodesic net $f$ is developable if and only if for every axis-aligned rectangle $P \subset \R^2$, the angles of the mapped curved rectangle $f(P)$ satisfy:
\begin{align}
\label{eq:angles_2pi}
\sum_{j=1}^{4} \alpha_j  = 2\pi.
\end{align}
\end{lemma}
\begin{proof}
Applying the local Gauss-Bonnet theorem to $P$ (see \cite{do_carmo}, chapter 4, page 268), we get
\begin{equation}
\label{eq:modGaussBonnet}
\int_{f(P)} K \, dA  + \int_{\partial f(P)} \kappa_g\, ds \ +\ \sum_{j=1}^4 \alpha_j   = 2\pi,
\end{equation}
where $K$ is the Gauss curvature and $\kappa_g$ is the geodesic curvature.
Since $f$ is a geodesic net, the images of $P$'s edges under $f$ are geodesics, and so $\kappa_g = 0$ on the curves of $\partial f(P)$, hence $\int_{\partial f(P)} \kappa_g\, ds = 0$.\\ 
$[\Rightarrow]$ Assume $f(P)$ is developable. Then $K$ vanishes and $\int_{f(P)} K \, dA = 0$, hence $\sum_{j=1}^4 \alpha_j = 2\pi$.\\
$[\Leftarrow]$ Assume $f(P)$ is not developable. Then there exists a point $p = f(x_*,y_*)$ such that $K(p) \neq 0$; assume w.l.o.g.\ $K(p) > 0$. There is a sufficiently small neighborhood $U$ with $(x_*,y_*) \in U$ such that $K > 0$ on $f(U)$. Let $P\subset U$ be an axis-aligned rectangle, then $\int_{f(P)} K\,dA > 0$ and from \eqref{eq:modGaussBonnet} we have $\sum_{j=1}^{4} \alpha_j  > 2\pi$, contradicting our condition \eqref{eq:angles_2pi}.
\end{proof}

\begin{topbot}
\begin{corollary}
\label{cor:cont_orth_geo_dev}
An orthogonal geodesic net $f$, i.e., a geodesic net with $\sphericalangle(t_1, t_2) = \frac{\pi}{2}$, is a developable net.
\end{corollary}
\end{topbot}
An isometry $f$ of a planar region $U \subseteq \R^2$ is an orthogonal geodesic net, as it maps a regular grid in the plane to orthogonal geodesics. Therefore the opposite is also true:  every developable net can be parameterized by an orthogonal geodesic net. This is summarized by the following corollary:

\begin{topbot}
\begin{corollary}
\label{cor:orth_geo_dev_equiv}
A smooth surface is developable if and only if it can be locally parameterized by orthogonal geodesics.
\end{corollary}
\end{topbot}
We are now ready to discuss discrete geodesic nets and our derivation of an equivalent condition for their orthogonality.

\subsection{Discrete geodesic nets}
As a base for our model we use the following definition:

\begin{mydefinition}
\label{def:wunderlich}
A discrete net $F$ is a \emph{discrete geodesic net} if each two opposing angles made by the edges of a star in the net are equal (see \figref{fig:geodesic_star}). 
\end{mydefinition}

This is a modification of a definition set by Wunderlich~\shortcite{wunderlich} in his work discretizing Voss surfaces, which are surfaces parameterized through conjugate geodesics. By \cite{wunderlich}, a discrete net $F$ is a discrete Voss surface if it is a planar quad net that also satisfies the angle condition in \defref{def:wunderlich}. We remove the planarity restriction, as we are interested in discretizing geodesics that are not necessarily conjugate.

To obtain an intuition, consider the polylines $(\Fbo, F, F_1)$ and $(\Fbt, F, F_2)$ as two discrete coordinate curves passing through point $F$. A geodesic curve is ``as straight as possible'', dividing the angle deviation from $\pi$ on both sides equally, i.e., $\alpha_1+\alpha_2 = \alpha_3 + \alpha_4$ for the first curve and 
$\alpha_2+\alpha_3 = \alpha_4 + \alpha_1$ for the second, where $\alpha_1, \ldots, \alpha_4$ are the angles around the star of $F$ (see \figref{fig:geodesic_star}). Together, these two conditions are equivalent to $\alpha_1=\alpha_3$ and $\alpha_2=\alpha_4$, as in \defref{def:wunderlich}.

\begin{figure}[b]
\centering
   \includegraphics[width=0.8\linewidth]{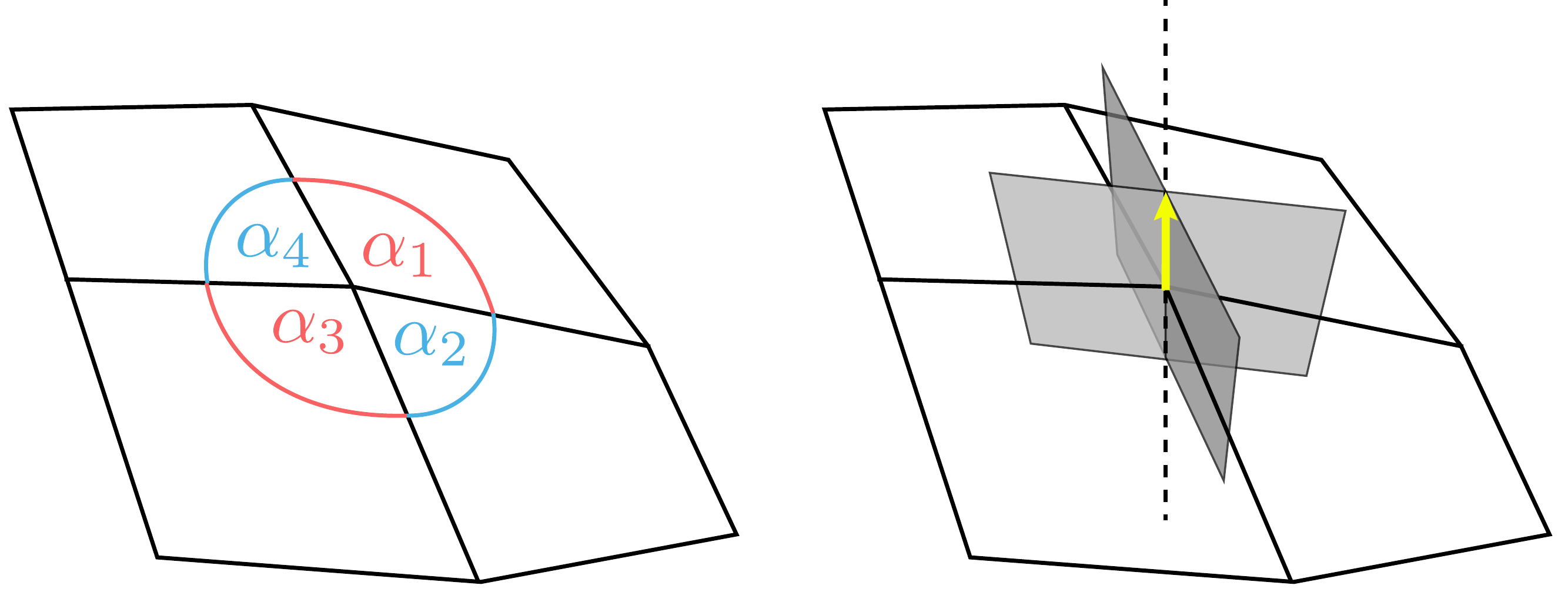}
   \caption{Left: A star in a discrete geodesic net has equal opposing angles. Right: On a geodesic star, the intersection of the osculating planes of the discrete coordinate curves is the surface normal.}
   \label{fig:geodesic_star}
\end{figure}

We define tangents and normals on discrete geodesic nets through their (discrete) coordinate lines, mimicking the properties of their continuous counterparts. On a smooth geodesic net $f$, let $p = f(x_0,y_0)$ be some point and $\gamma_1(t) = f(x_0 + t,y_0),\  \gamma_2(t) = f(x_0,y_0 + t)$ the coordinate lines through $p$.  The curves $\gamma_1$ and $\gamma_2$ are geodesics emanating from $p$ at two linearly independent directions $\gamma_1'(0) = f_x,\  \gamma_2'(0) = f_y$. If $\gamma_1(t)$ is regular and non-degenerate at $0$, i.e., $\gamma_1'(0), \gamma_1''(0) \neq 0$, it has a well defined Frenet frame $\{t_1,n_1,b_1\}$ and an osculating plane $\Pi_1$ spanned by $t_1,n_1$. Since $\gamma_1(t)$ has zero geodesic curvature, its curvature is equal to the normal curvature of the surface, which implies that the curve's normal is in fact parallel to the surface normal at $p$: $n_1 \parallel n$ (where $n = \frac{t_1 \times t_2}{\norm{t_1 \times t_2}}$). If also $\gamma_2$ has non-vanishing first and second derivatives, the surface normal $n$ is parallel to the intersection line between the two osculating planes $\Pi_1, \Pi_2$. We can find a natural discrete model for those quantities for a discrete geodesic net $F$.

Let $F$ be a vertex on a discrete geodesic net, and let $\Gamma_1, \Gamma_2$ be discrete geodesic curves through $\Fbo,F,F_{1}$ and $\Fbt,F,F_{2}$, respectively. We say that the curve $\Gamma_j$ is non-degenerate if the three points $F_{\bar j},F,F_{j}$ are not collinear. In that case, we can define the osculating plane and Frenet frame:

\begin{mydefinition} \label{def:discrete_osc_plane}  The osculating plane \ $\Pi_j, \ j=1,2,$ of a non-degenerate discrete curve \ $\Gamma_j$ through vertices $F_{\bar j},F,F_{j}$ is the plane passing through these three points.
The Frenet frame of \ $\Gamma_j$ at $F$ is denoted by \ $\{T_j,N_j,B_j\}$, where 
\begin{eqnarray*}
 &T_j = \frac{\delta_{j}F- \delta_{\bar j}F}{\norm{\delta_{j}F - \delta_{\bar j}F}},\ 
 N_j = \frac{\delta_{j}F + \delta_{\bar j}F}{\norm{\delta_{j}F + \delta_{\bar j}F}},\ 
 B_j = T_j \times N_j.
\end{eqnarray*}
\end{mydefinition}
See \figref{fig:discrete_frenet_frame} for an illustration. Note that $T_j$ are well defined also when $F_{\bar j},F,F_{j}$ are collinear, and are never zero as our net is assumed to be a discrete immersion.

\begin{figure}[h]
\centering
   \includegraphics[width=0.6\linewidth]{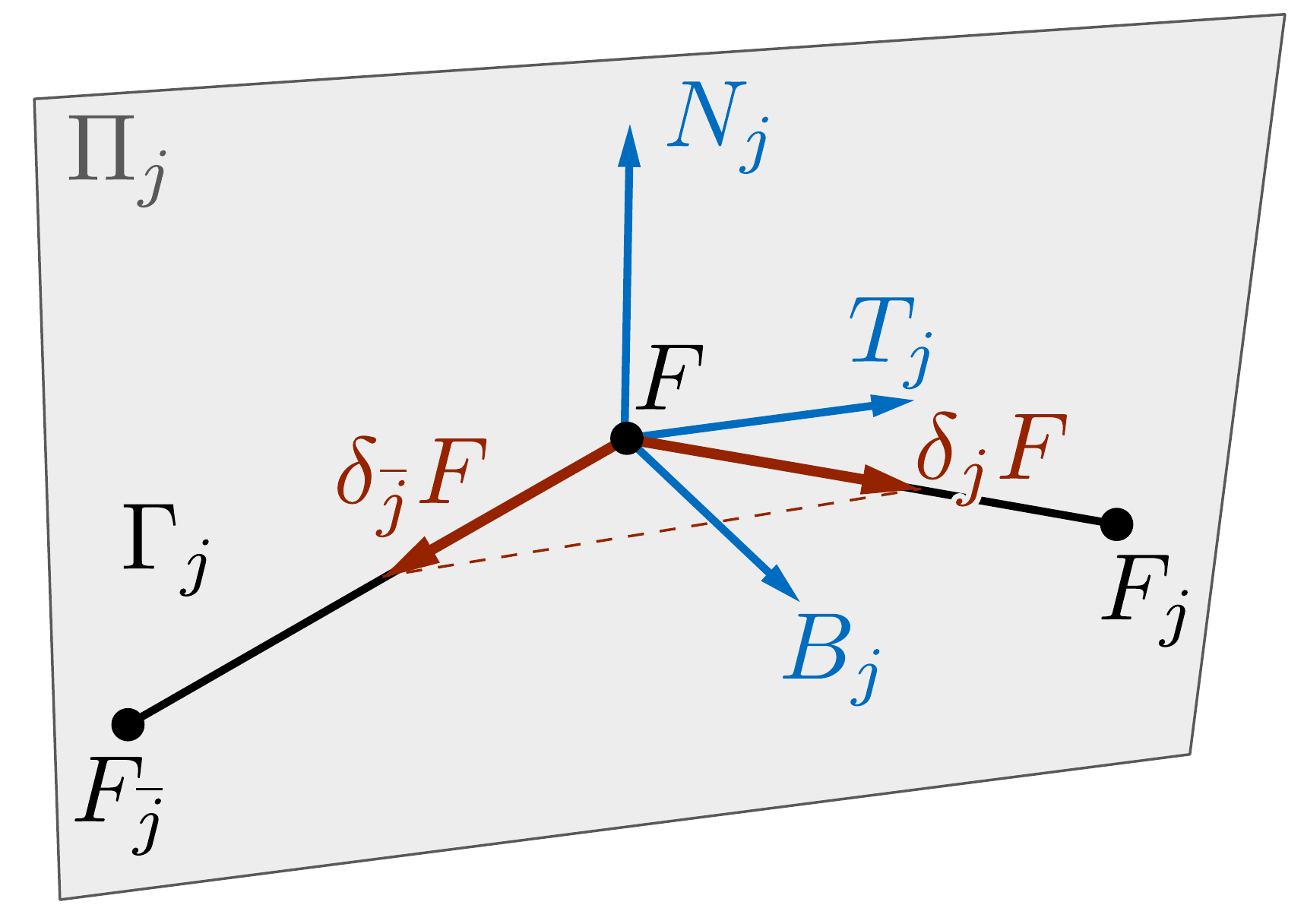}
   \caption{A discrete coordinate curve $\Gamma_j$ at $F$ (in black), its osculating plane $\Pi_j$ spanned by the edges of $\Gamma_j$, the Frenet frame (in blue): tangent $T_j$, normal~$N_j$  and binormal $B_j$. The dashed red vector is $\delta_jF-\delta_{\bar j}F$.}
   \label{fig:discrete_frenet_frame}
\end{figure}

\begin{mydefinition} 
\label{def:N_map} 
The discrete Gauss map of a geodesic net $F$ is $$N=\frac{T_1 \times T_2}{\norm{T_1 \times T_2}},$$
where $T_1, T_2$ are defined as above.
\end{mydefinition}

Just as in the continuous case, the principle normals of discrete geodesic curves and the surface normal agree, as shown by the following lemma:

\begin{lemma}\label{Lem:discrete_geo_n}
Let $\Gamma_1, \Gamma_2$ be two non-degenerate discrete curves around a vertex of a discrete geodesic net $F$ and $\{T_1,N_1,B_1\}, \{T_2,N_2,B_2\}$ their discrete Frenet frames. Then $N_1, N_2$ and the discrete surface normal $N$ (see \defref{def:N_map}) are all parallel and lie on the intersection of the osculating planes $\Pi_1$ and $\Pi_2$.
\end{lemma}
\begin{proof}
By construction, $N_1 \bot T_1 $, and by direct computation using the opposite angles condition (\defref{def:wunderlich}) we have $\langle N_1, T_2 \rangle  = 0$. Therefore $N_1 \parallel 
 N$. Similar computation shows $N_2 \bot T_1$ and therefore $N_2 \parallel N$. \end{proof}
Note that $N$ is the angle bisector of both discrete curves meeting at $F$, see \figref{fig:geodesic_star}.

\subsection{Discrete developable geodesic net}
Using the tangents defined above, we are now ready to define discrete developable surfaces through nets of orthogonal geodesics:
\begin{topbot}
\begin{mydefinition} 
\label{def:orth_geo_net}  
A discrete orthogonal geodesic net is a discrete geodesic net where at every star, the discrete tangents of the two discrete coordinate curves are orthogonal: $T_1  \bot T_2 $. Such a net is a discrete developable surface in orthogonal geodesic parameterization. 
\end{mydefinition}
\end{topbot}
This definition obviously reflects the smooth case, where an existence of an orthogonal geodesic net on a surface is equivalent to developablity (\corollaryref{cor:orth_geo_dev_equiv}). The following theorem provides useful interpretations of our net and helps to see why this definition is equivalent to \defref{def:def1} (see also \figref{fig:geo_cross}). 

\begin{figure}[b]
\centering
   \includegraphics[width=\linewidth]{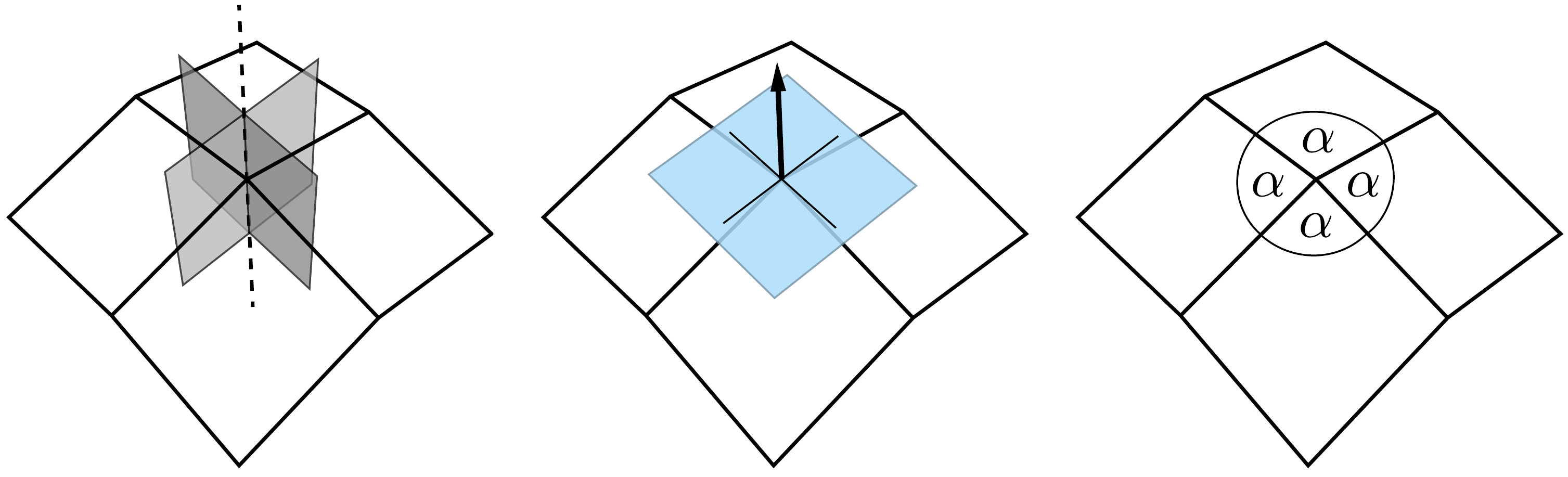}
   \caption{These three conditions on a geodesic star are equivalent: the two osculating planes are perpendicular to each other (left), the projection of the star's edges onto the tangent plane forms an orthogonal cross (middle), all angles around the star are equal (right).}
   \label{fig:geo_cross}
\end{figure}

\begin{figure*}[t]
\includegraphics[width=\linewidth]{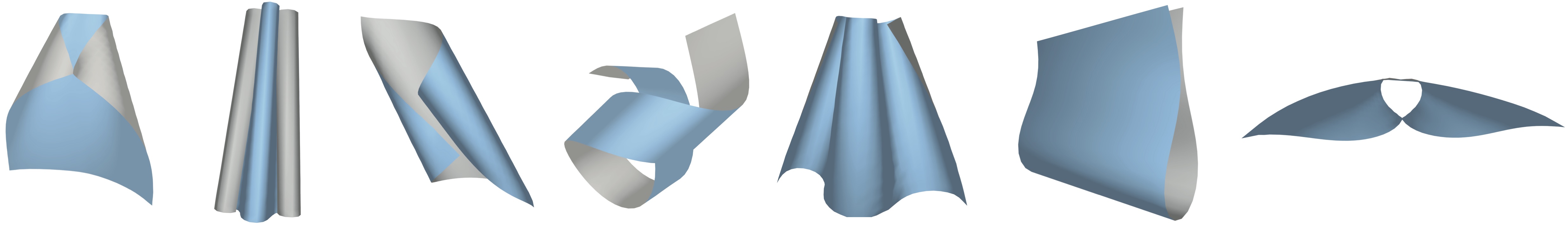}
\caption{\label{fig:paper_folding} Developable surfaces created using our vertex-handle based editing system. These examples were designed by deforming a flat sheet.}	
\end{figure*}

\begin{theorem}\label{Thm:orth_geo_equiv}
Assume a star has equal opposing angles, i.e., it fulfills the angles condition for discrete geodesic nets (\defref{def:wunderlich}). Then the following conditions are equivalent:
\begin{enumerate}
\item The discrete tangents of the coordinate curves are orthogonal: $T_1 \bot T_2$. \label{orth_tangents}
\item The edges of the star form a right-angle cross when projected into the discrete tangent plane, which is the plane orthogonal to the discrete normal $N$.\label{right_cross}
\item All angles between consecutive edges of the star are equal. \label{orth_geo_angle_cond}
\end{enumerate}
\end{theorem}
\begin{proof}
By \lemmaref{Lem:discrete_geo_n}, $N$ is a bisector of the unit-length vectors $\delta_{1}F,\delta_{\bar 1}F$, as well of the unit-length vectors $\delta_{2}F,\delta_{\bar 2}F$. Adding the formulas for $T_1, T_2$ from \defref{def:discrete_osc_plane} we have
\begin{align*}
\delta_1F + \delta_{\bar 1}F &= \tilde{a}\, N, &\quad \delta_2F + \delta_{\bar 2}F &= \tilde{c}\, N, \\
\delta_1F - \delta_{\bar 1}F &= \tilde{b}\,T_1,   &\quad \delta_2F - \delta_{\bar 2}F &= \tilde{d}\,T_2,
\end{align*}
for some $\tilde{a},\tilde{b},\tilde{c},\tilde{d} \in \R$. By adding/subtracting the respective equations of the second row to/from the first row, we can write the star's edge directions as
\begin{align*}
\delta_{1}F &= a\,N + b\, T_1, &\quad  \delta_{2}F &= c\,N + d\, T_2,\\
\delta_{\bar 1}F &= a\,N - b\,T_1, &\quad \delta_{\bar 2}F &= c\,N - d\, T_2,
\end{align*}
for some $a,b,c,d \in \R$. \\
$[\,$\eqref{orth_tangents}$\!\iff\!$\eqref{right_cross}$\,] $  Projection to the tangent plane is equivalent to removing the normal component $N$ from each vector, hence the direction vectors of the projected star edges are $bT_1, -bT_1, dT_2, -dT_2$ and the claim follows. \\
$[\,$\eqref{orth_geo_angle_cond}$\!\iff\!$\eqref{orth_tangents}$\,]$
As we assume opposing angles in the star are equal, \eqref{orth_geo_angle_cond}$\iff\! \langle \delta_{1}F,\ \delta_{2}F \rangle = \langle \delta_{\bar 1}F,\ \delta_{2}F \rangle \!\iff\! \langle aN + b T_1,\ cN + d T_2 \rangle = \langle aN - bT_1,\ cN + d T_2 \rangle$, which is equivalent to $T_1 \bot T_2$ for a non-degenerate star.
\end{proof}
Note that the third condition (all angles in the star are equal) subsumes the condition for a discrete geodesic net (\defref{def:wunderlich}) and conveniently encapsulates discrete orthogonal geodesic nets, as we expressed in \defref{def:def1}.

%% file: 05-p-modeling-dev-surfaces.tex

\section{Modeling deformations of discrete developable surfaces}
\label{sec:editing_system}
Our definition of discrete developable surfaces (\defref{def:def1}) is simple and local, such that it can be easily used in applications. We demonstrate this in an interactive editing system for discrete developable surfaces. Starting from a given discrete orthogonal geodesic net $F^0$, e.g., an orthogonal planar grid or a cylinder, the user can fix and move vertices around, as well as glue together or sever vertices. The latter is permitted only in case the operation keeps the mesh a (not necessarily oriented) manifold. We denote the set of vertices manipulated by the user (the handles) by $\mathcal{H}$. Whenever the user moves the handle vertices, the system computes a result from the space of discrete orthogonal geodesic nets, which is as close as possible to the prescribed handle positions. We analyze this shape space in \secref{sec:rigidity}. To choose a \emph{good}, or intuitive solution, our optimization includes isometry and smoothness regularizers, as well as constraints for boundary vertices. 


\subsection{Orthogonal geodesic constraints} \label{opt_const}
\defref{def:def1} gives us the feasible shape space through a set of constraints on each inner vertex of $F$ and a generalization for boundary vertices. We constrain every vertex to have all its corner angles equal. Let $e_j,\ j=1,\ldots,l,$ be the set of edges originating at a vertex $v$, ordered such that consecutive edges share a quad. Then the condition $\sphericalangle (e_j, e_{j+1}) = \sphericalangle (e_{j+1}, e_{j+2})$ is equivalent to:
\begin{equation}
\label{eq:star_constraints}
\langle e_j, e_{j+1} \rangle \|e_{j+2}\| - \langle e_{j+1}, e_{j+2} \rangle \|{e_{j}}\| = 0.
\end{equation}
In case of a corner boundary vertex with only two incident edges $e_1$ and $e_2$ and one angle, we constrain the angle to remain as in the reference shape: 
\begin{equation}
\label{eq:corner_constraints}
\frac{\langle e_1, e_2 \rangle}{\norm{e_1} \norm{e_2}} - \arccos(\alpha)  = 0,
\end{equation}
where $\alpha = \sphericalangle(e_1, e_2)$ in $F^0$. We denote the constraints \eqref{eq:star_constraints}, \eqref{eq:corner_constraints} as $c_i(F) = 0, \ i = 1, \ldots, m,$ where $i$ enumerates all the inner and boundary vertices and their relevant incident edges.

\subsection{Smoothness and isometry regularizers}
The constraints above do not encode smoothness or isometry, and simply projecting a given initial guess onto the feasible space might lead to unintuitive results. To generate smooth and aesthetically pleasing deformations, we seek a feasible solution that minimizes a deformation energy $E(F)$.
We employ a simple smoothness term, namely the Laplacian energy of the displacement w.r.t.\ the current state of the shape, or the current ``frame'', $F^k$:
\begin{equation}
E_\textrm{smooth}(F) = \norm{L(F) - L(F^k)}^2,
\end{equation}
where we use the simple uniform Laplacian $L$. The second energy term encourages maintaining isometry of the boundary, intuitively helping to control the scaling of the deformation:
\begin{equation}
E_\textrm{iso}(F) = \sum_{e_j \in \partial F} (\|{e_j}\|-l_j)^2,
\end{equation}
where $\partial F$ is the set of boundary edges of $F$, and $l_j$'s are the edge lengths in $F_0$. 
Finally, we add the positions of the handle vertices as soft constraints, since the user is likely to manipulate the handles in ways that are at odds with the developability constraints. The overall deformation energy is therefore
\begin{equation}
E(F) = E_\textrm{smooth} + w_\textrm{iso} E_\textrm{iso}(F) + w_\textrm{pos} \sum_{v\in \mathcal{H}}\norm{v - v_c}^2,
\end{equation}
where $v_c$ are the handle positions prescribed by the user and $w_\textrm{iso}$, $w_\textrm{pos}$ are scalar weights.

\subsection{Optimization} \label{sec:opt}
In each frame, we solve the following optimization problem:
\begin{equation} \label{eq:const_opt}
\begin{aligned}
& \argmin_F
& & E(F) \\
& \textrm{subject to}
& & c_i(F) = 0, \ \  i = 1, \ldots, m.
\end{aligned}
\end{equation}
We use the quadratic penalty method \cite{Nocedal}, which converts the above constrained minimization to a series of unconstrained problems of the form
\begin{equation} \label{inner_iter}
\argmin_F \quad w\, E(F) + \sum_{i} c_i (F)^2.
\end{equation}
The above is iterated starting with $w=w_0$ and halving the weight $w$ in each subsequent iteration, until the constraints are satisfied numerically, i.e.\ $\sum_{i} c_i (F)^2 < \epsilon$. The minimizations \eqref{inner_iter} are solved using using L-BFGS \cite{lbfgs}, where we use ARAP \cite{arap} with the given positional constraints to get an initial guess. The figures in this paper and the accompanying video were generated with the parameters $w_0=1,\ w_\textrm{iso}=1,\ w_\textrm{pos}=0.1,\ \epsilon = 1\text{e}{-12}$,  and the input mesh was first scaled to have an average edge length of $1$.

\begin{figure}[t]
\includegraphics[width=0.9\linewidth]{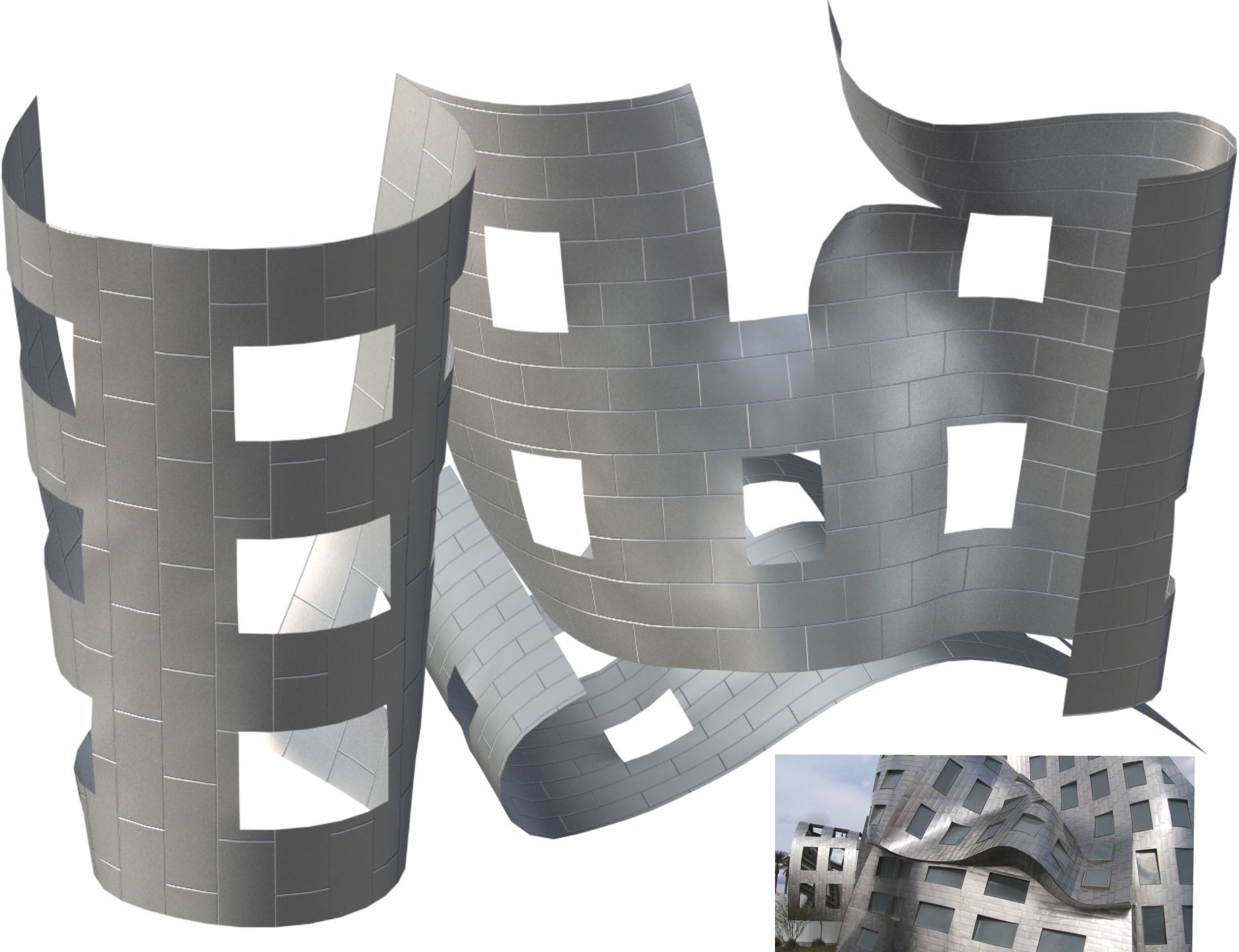}
\caption{\label{fig:gehry} An example of developable shapes with nontrivial topology, created in our interactive editing system (rendering done offline). Inspired by Frank Gehry's design of the Lou Ruvo Center for Brain Health (lower right). The texture coordinates for our model are simply obtained from the vertex coordinates of a planar rectangular grid with the same boundary edge lengths.}	
\end{figure}

\subsection{Results}
We implemented our editing system on a 3.4 GHz Intel Core i7 machine, on which our single threaded implementation can handle around 1000 vertices interactively.
The results in \figref{fig:paper_folding} demonstrate a variety of rolled, paper-like shapes similar to the results of \cite{solomon}, but made with a more intuitive, vertex-handle based editing system (see also the accompanying video). Our system can seamlessly handle surfaces with nontrivial topology, as well as non-orientable surfaces, as shown in Figs.\ \ref{fig:teaser}, \ref{fig:lantern}, \ref{fig:gehry}. 

Similarly to other nets in DDG, e.g., discrete $K$-surfaces, the geometric information of our net is only the vertex positions. Edges should not be seen as part of the surface, and the non-planarity of the quads in our model implies that we can only render and fabricate our surfaces by arbitrarily triangulating them. Note that this would also be the case for a dense sampling of a general smooth orthogonal geodesic net, which approximates our model, as shown later in \secref{taylor_sampling}. 
Nevertheless, we demonstrate in \figref{fig:fabrication} that our discrete model could be used for fabrication purposes.

\begin{figure}
\includegraphics[width=0.5\linewidth]{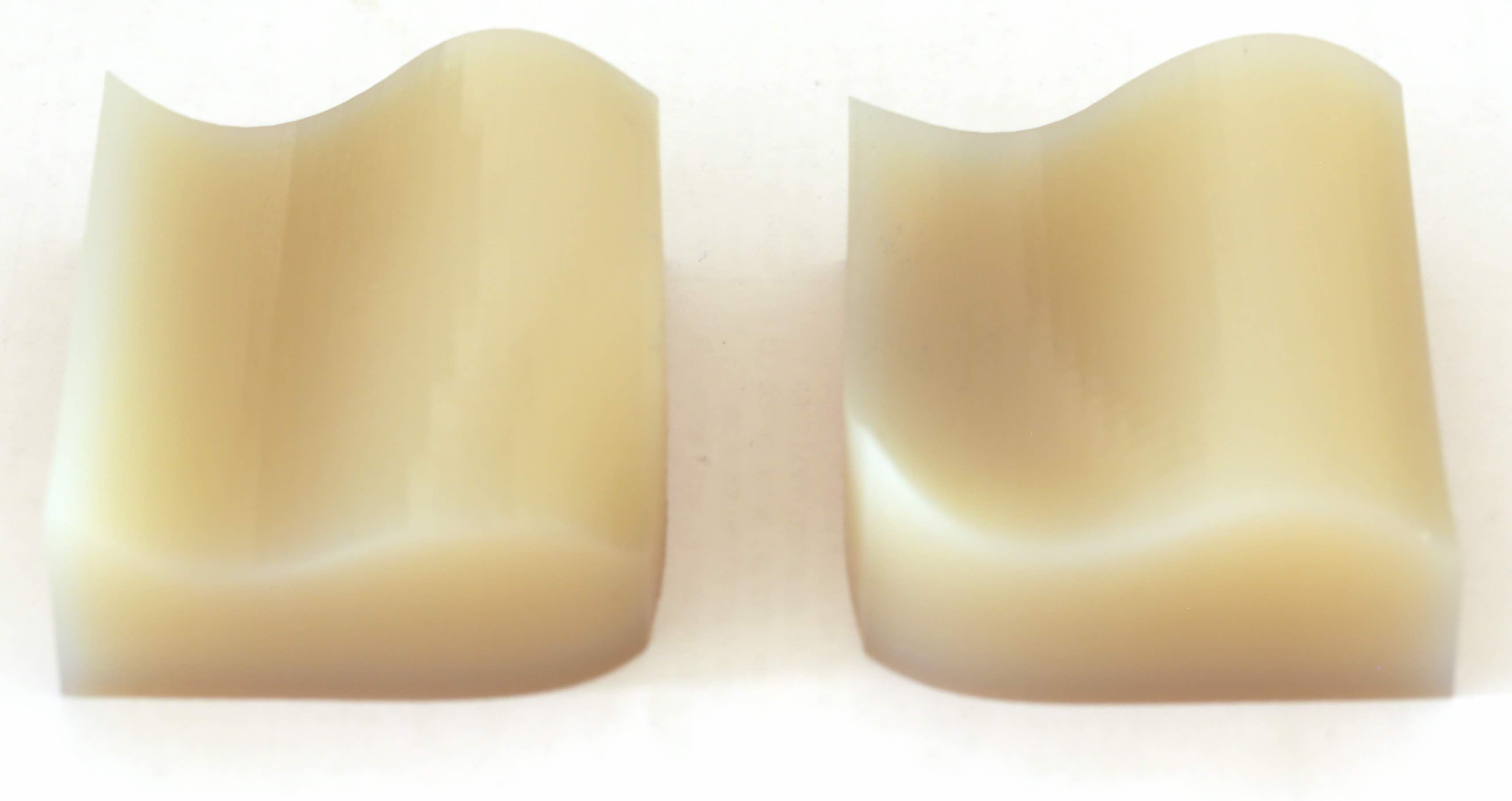}\includegraphics[width=0.5\linewidth]{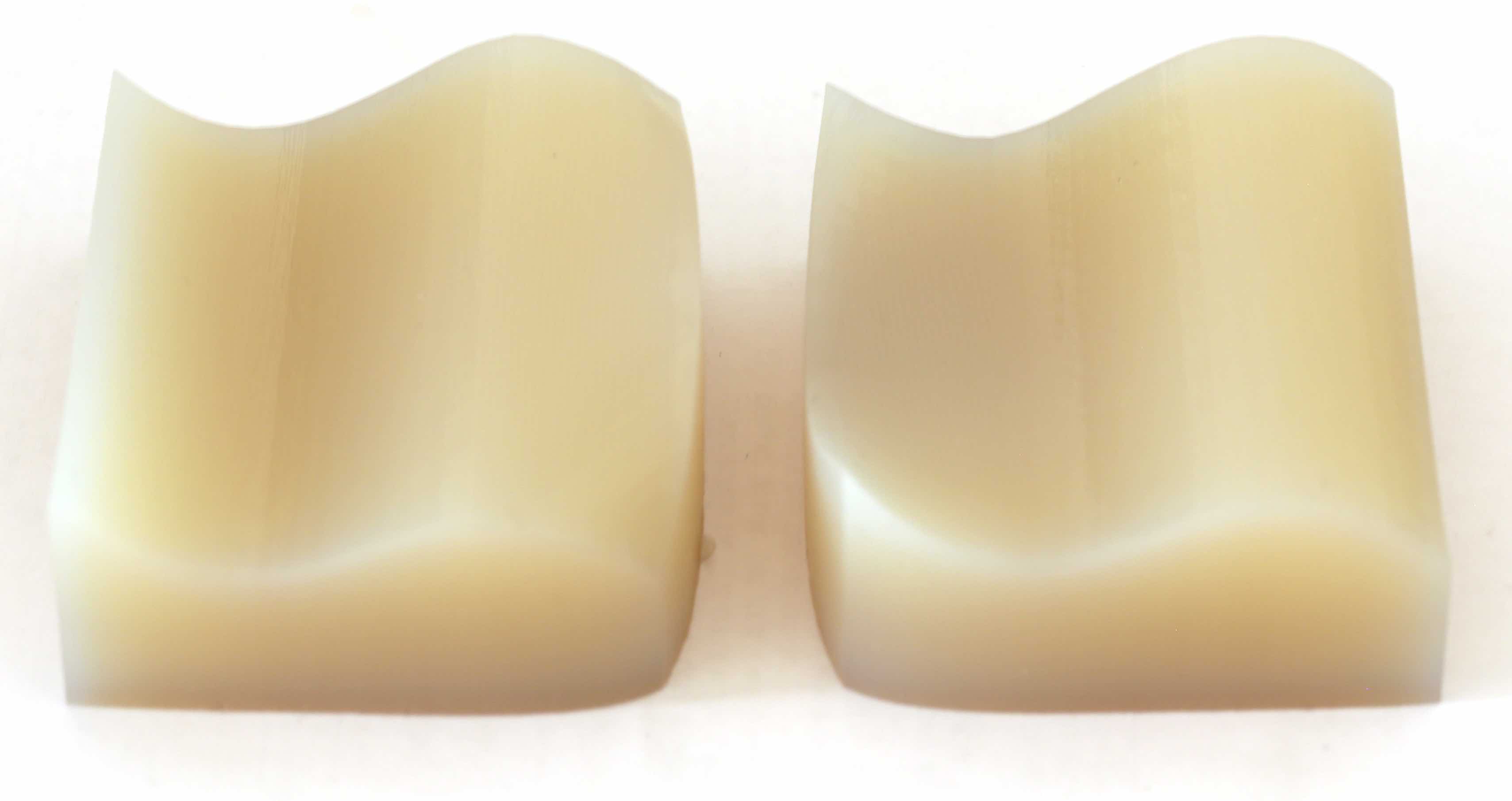}\\
\includegraphics[width=0.5\linewidth]{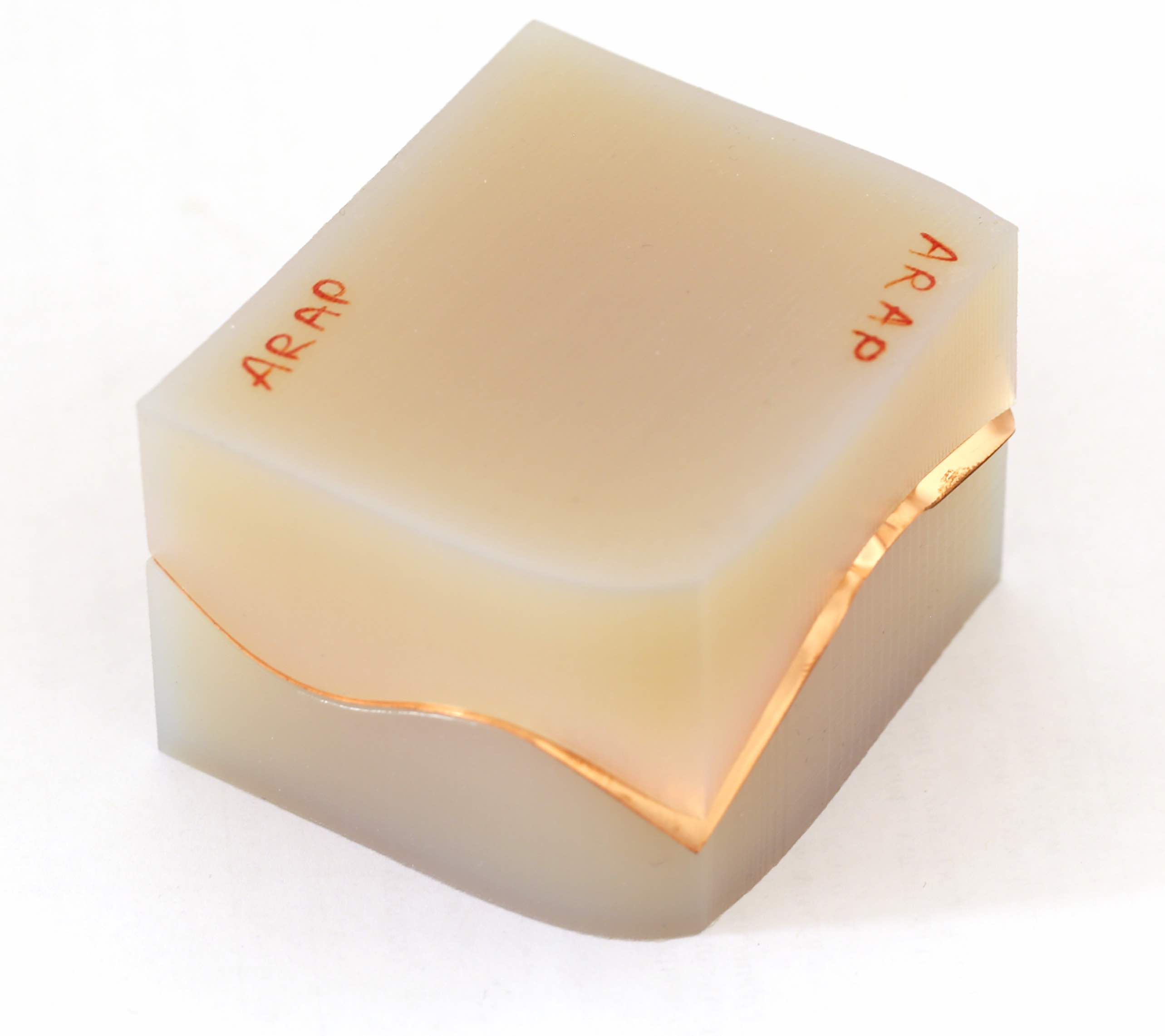}\includegraphics[width=0.5\linewidth]{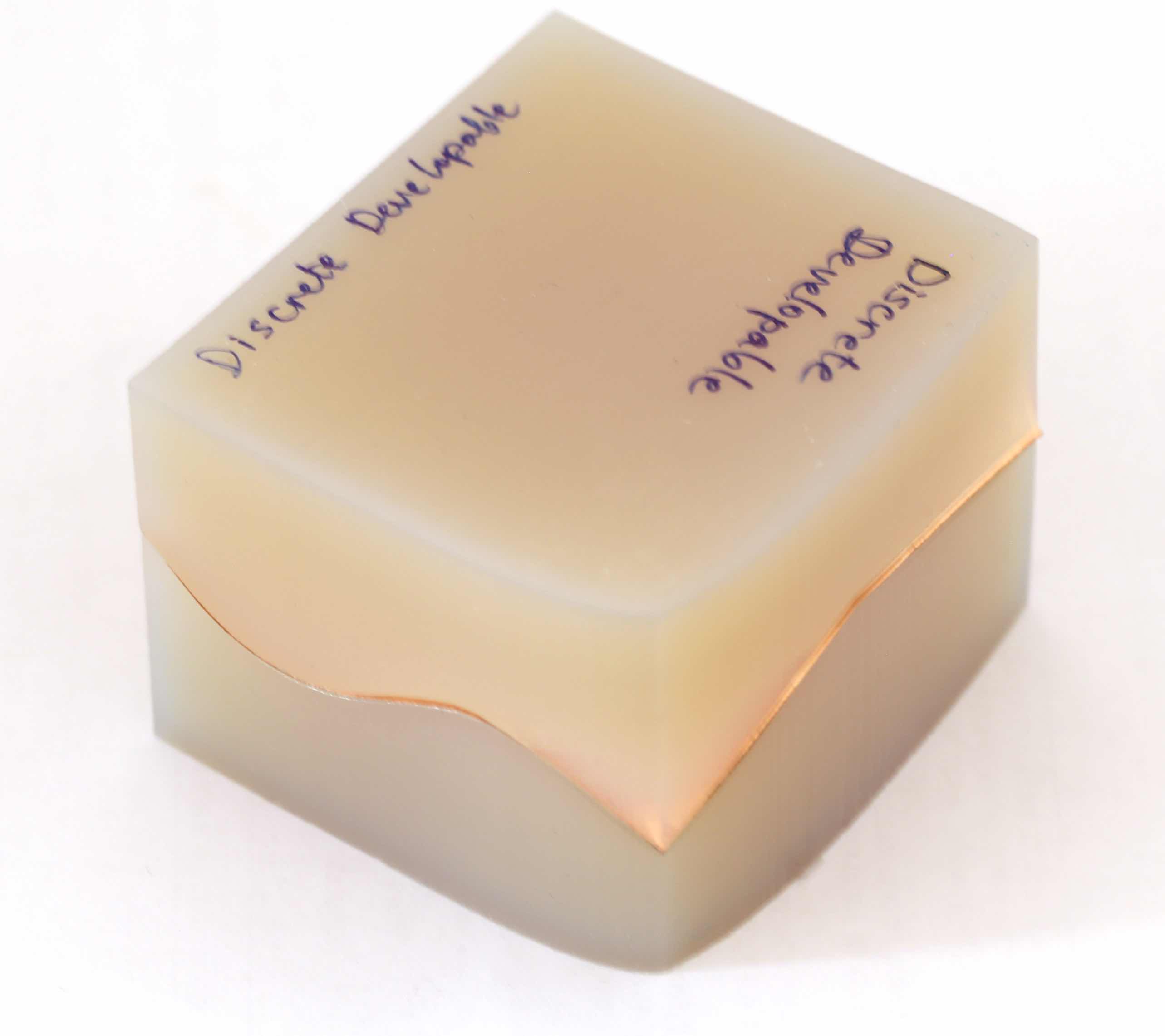}\\
\includegraphics[width=0.5\linewidth]{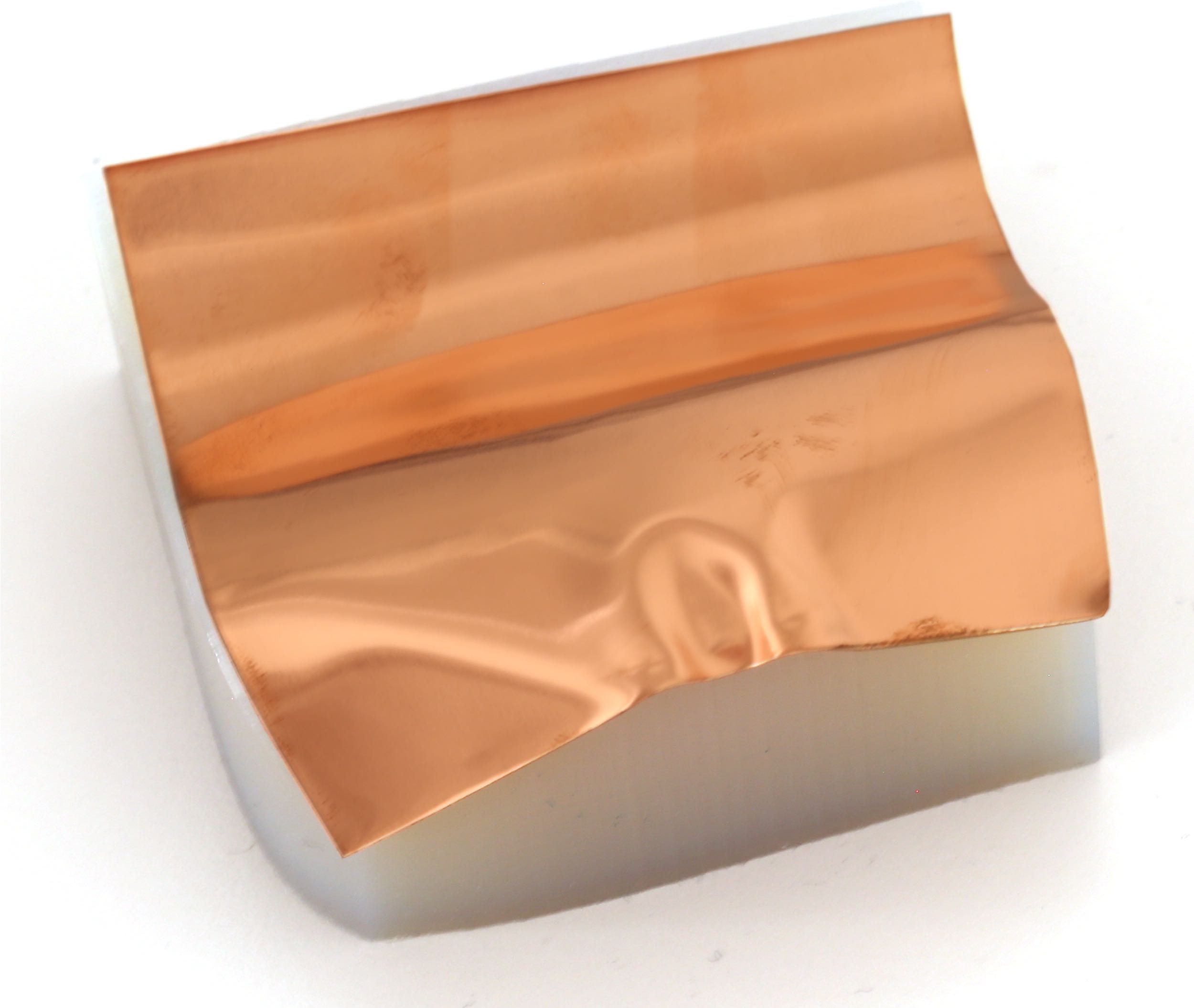}\includegraphics[width=0.5\linewidth]{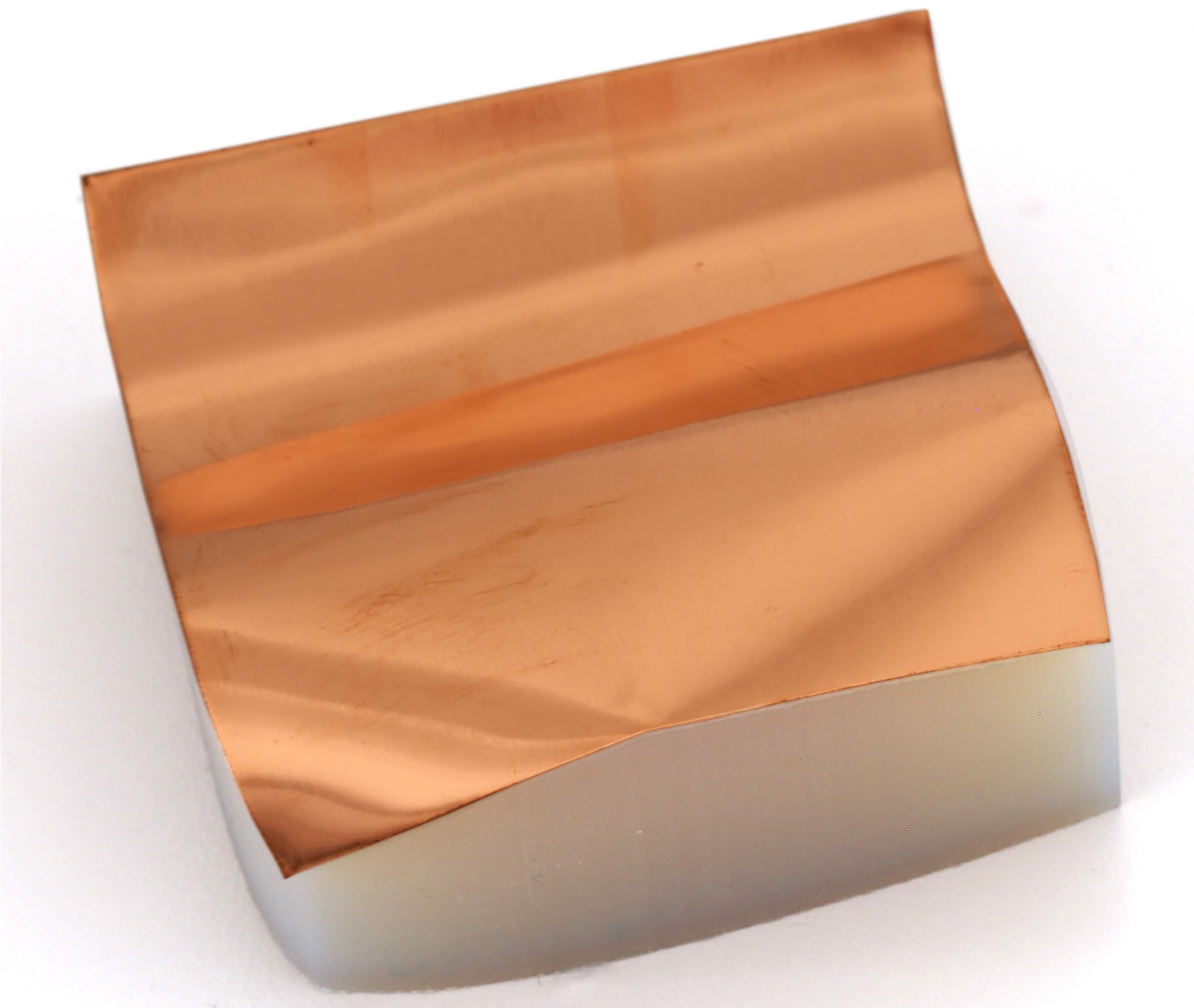}
\caption{\label{fig:fabrication}Validation by fabrication. We 3D-printed a ``sandwich'' (top row) whose inner cut surface is the result of deforming a flat square using ARAP (left column) and our editing system for discrete developable surfaces (right column) using the same (soft) positional constraints. We cut two squares out of a thin copper sheet with the dimensions of the initial model before deformation, and we sandwiched these squares in the fabricated pieces (second row). 
Not surprisingly, since the ARAP result is not developable, the sheet wrinkles and buckles (bottom left), while our result exhibits pure bending (bottom right). Note that our result is smooth everywhere except at the boundary, where a cone-like kink is created in the digital model to remedy the doubly curved ARAP surface.}
\end{figure}

%% file: 05-q-parallels-to-smooth.tex

\section{Normals and rulings} \label{sec:normals_and_rulings}

We continue investigating our discrete developable surface model by looking at the Gauss map and a simple local definition of the rulings. Although our model does not explicitly enforce any properties of these two objects, we empirically see that their behavior corresponds well to the expected properties of a developable surface.

\subsection{One-dimensional Gauss map} \label{gauss_dim}
In the continuous case, a smooth developable surface $f$ has vanishing Gaussian curvature. Since it corresponds to the area of the Gauss map, it means that the normal map $n$ of $f$ is one-dimensional \cite{do_carmo}. \defref{def:N_map} supplies us with a discrete per-vertex normal on a discrete geodesic net $F$, and we can view the collection of all vertex normals with the connectivity of $F$ as a discrete net $N$. 
%
%
We show in \figref{fig:gauss_map} and the supplementary video that editing with our system results in a discrete Gauss map that is approximately one-dimensional.

\begin{figure}[h]
\centering
   \includegraphics[width=\linewidth]{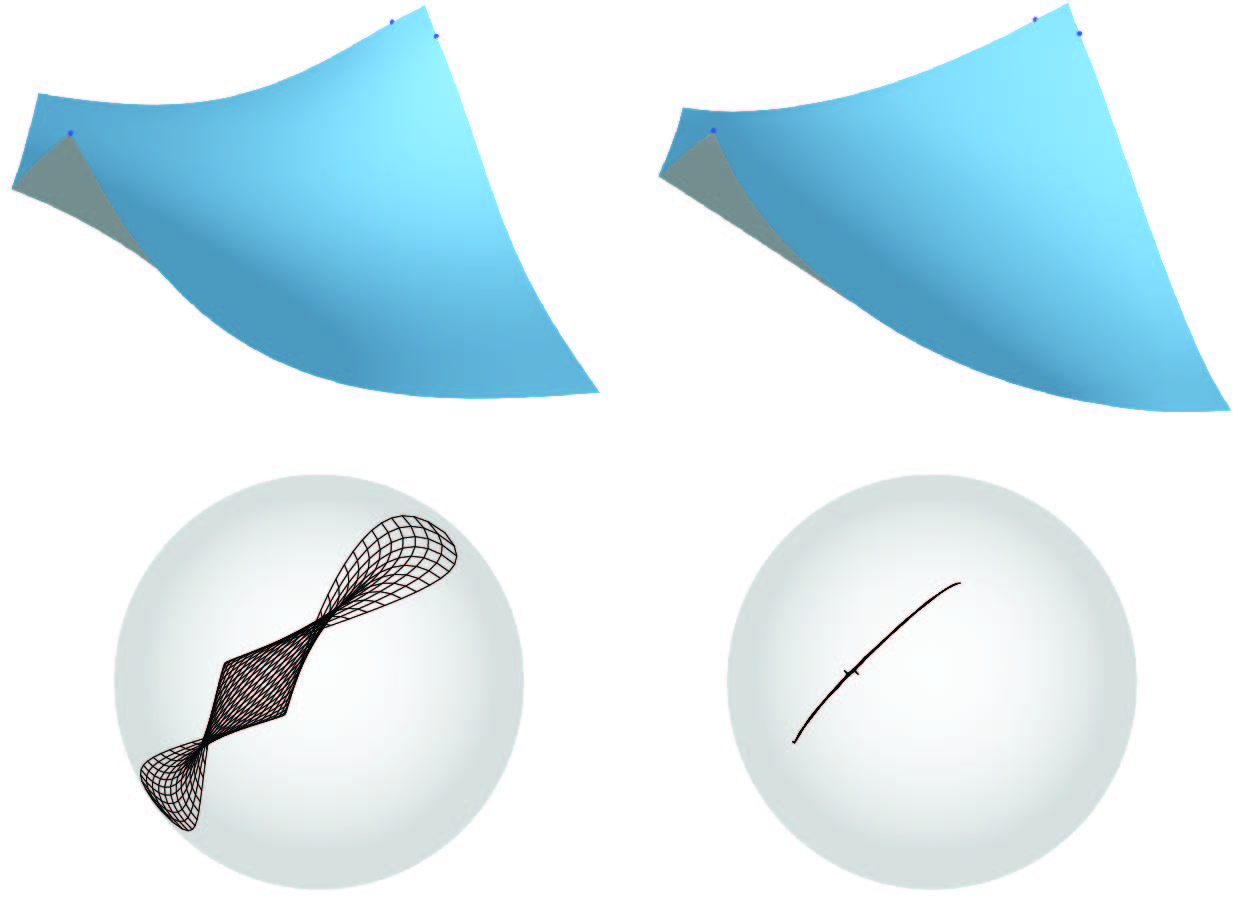}
   \caption{Left: The result of moving the lower left corner of a planar mesh towards the upper right corner using ARAP deformation~{\protect\cite{arap}}. Right: the same positional constraints are employed in our developable surface deformation system (\secref{sec:editing_system}), using the result of ARAP on the left as the initial guess. In this case the soft positional constraints are satisfied up to high precision. Below each net we display the image of its Gauss map $N$ (which in this case is virtually indistinguishable from the standard vertex-based normals of a triangle mesh obtained by triangulating the quad net). Note that our Gauss map tends to be one-dimensional.}
   \label{fig:gauss_map}
\end{figure}

\subsection{Vertex based rulings} \label{rulings_approx}
 Intuitively, rulings are line segments on a surface generated by the intersection of infinitesimally close tangent planes. As mentioned above, the Gauss map $n$ of a smooth developable net $f$ has a one-dimensional image, or, equivalently, parallel partial derivatives: $n_x \parallel n_y$. There is a unique ruling emanating from every non-planar on the surface in a direction $r$ that is orthogonal to $n$. The ruling is a curvature line, hence it is also orthogonal to the other principal direction $n_x \parallel n_y$ \cite{do_carmo}.
 Therefore, if w.l.o.g.\ $\langle n_x, n_y \rangle \geq 0$, then $r \parallel n \times (n_x + n_y)$. This holds even if one of the terms $n_x,\ n_y$ vanishes. This can be readily discretized:
 	\begin{mydefinition} \label{def:disc_rulings}  The direction of a discrete ruling, emanating from a point $F$ of a discrete geodesic developable net is 
	$$R = N \times (N_x + N_y),$$ 
	where $N_x = N_1-N_{\bar 1}$ and $N_y = N_2 - N_{\bar 2}$, oriented such that \mbox{$\langle N_x, N_y \rangle \geq 0$}. \end{mydefinition}
 	\defref{def:disc_rulings} is entirely local, however in practice the discrete rulings tend to fit the surface globally, see \figref{fig:rulings}. Note that the definition above is only valid at inner vertices with all neighbors being inner vertices as well, such that $N_x,N_y$ are defined. Unlike in the continuous case, $N_x$ and $N_y$ are not necessarily parallel.

\section{Analysis and parallels with the smooth model} \label{sec:rigidity}
In this section we further study discrete geodesic nets, drawing parallels between the discrete and continuous cases. We analyze the variety of shapes that can be modeled by discrete orthogonal geodesic nets given in \defref{def:def1}. Loosely speaking, a \emph{good} discrete developable model should be sufficiently flexible to approximate every smooth developable surface, which we show by the Taylor expansion analysis in \secref{taylor_sampling}. The model should also be sufficiently restrictive, or rigid, to avoid unreasonable shapes. To that end, in  \secref{dev_extension} we show that our discrete orthogonal geodesic nets share a similar rigid behavior with a smooth developable surface.
In \secref{curve_geo_thm} we prove a discrete analogue for a simple theorem connecting curvature line nets, geodesic nets and orthogonal geodesic nets.

\begin{figure}[b]
\centering
   \includegraphics[width=\linewidth]{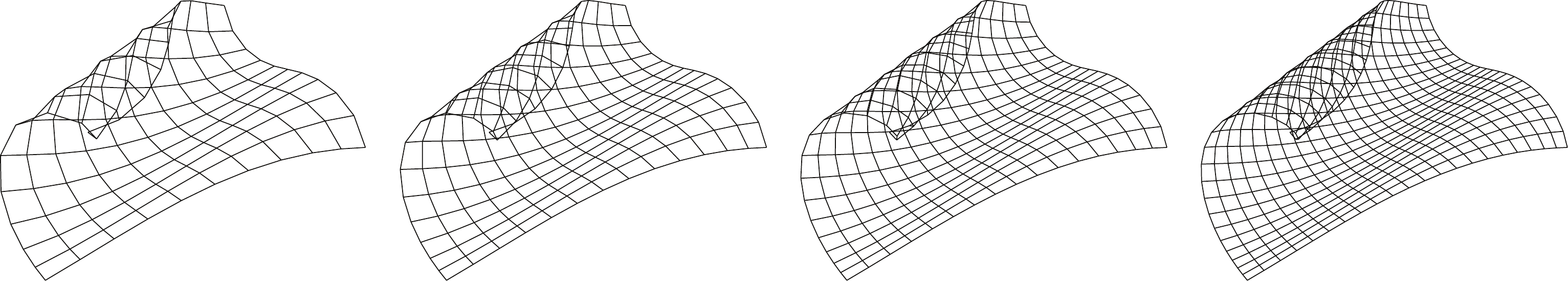}
   \caption{A series of samplings of a smooth orthogonal geodesic net $f$ with increasing sampling density. By \theoremref{Thm:taylor_orth_geo}, the stars of these discrete nets have equal angles up to second order, hence a discrete orthogonal geodesic net $F$ can also be viewed as an approximate sampling of a smooth orthogonal geodesic net $f$.}
   \label{fig:taylor}
\end{figure}

 	\begin{figure*}[t]
\centering
   \includegraphics[width=\linewidth]{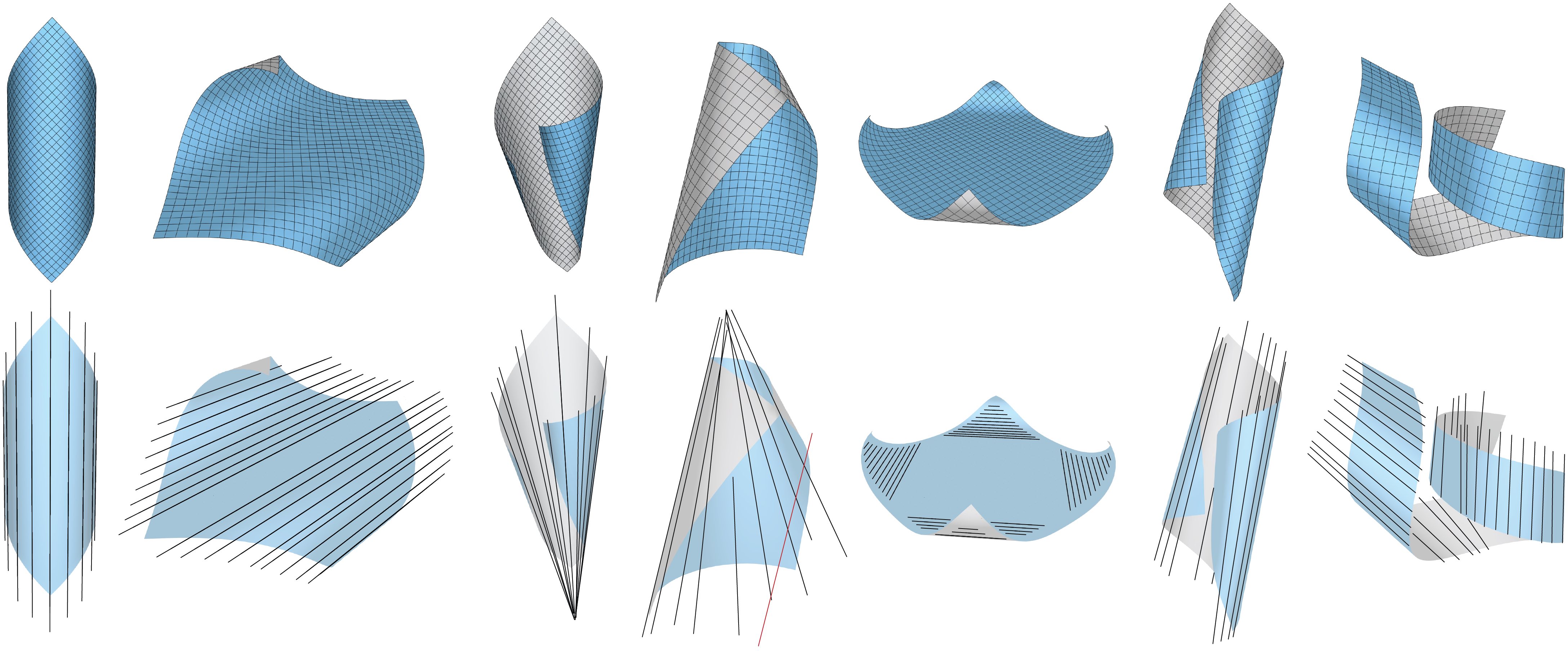}
   \caption{Discrete developable geodesic nets and their vertex based rulings. From left to right: Two cylinders, two cones, a planar region connected to four cylinders, two tangent developable surfaces. Remarkably, the rulings tend to fit the shapes globally, despite their entirely local definition (\defref{def:disc_rulings}). They are, however, less stable around planar regions due to the calculation of $N_x,N_y$ (see the red ruling in the center). We visualize the rulings sparsely for clarity. See the accompanying video for a three dimensional view.}
   \label{fig:rulings}
\end{figure*}

\subsection{Approximation of an analytical, smooth orthogonal geodesic net} 
\label{taylor_sampling}
%
Let $f$ be an arbitrary analytical smooth net and $p = f(x,y)$ a point on the surface. Imagine sampling points around $p$ to generate a discrete star. We show that this star is a discrete orthogonal geodesic star as in \defref{def:def1} up to second order if and only if $f$ is an orthogonal geodesic net (\figref{fig:taylor}).

 Let $\epsilon > 0$ and let
$f = f(x,y),\ f_1(\epsilon) = f(x+\epsilon,y),\ f_{\bar1}(\epsilon) = f(x-\epsilon,y),\ f_2(\epsilon) = f(x,y + \epsilon),\ f_{\bar2}(\epsilon) = f(x,y - \epsilon)$. 

\setlength{\intextsep}{8pt}%
\setlength{\columnsep}{8pt}%
\begin{wrapfigure}{r}{0.4\linewidth}
  \centering
  \includegraphics[width=\linewidth]{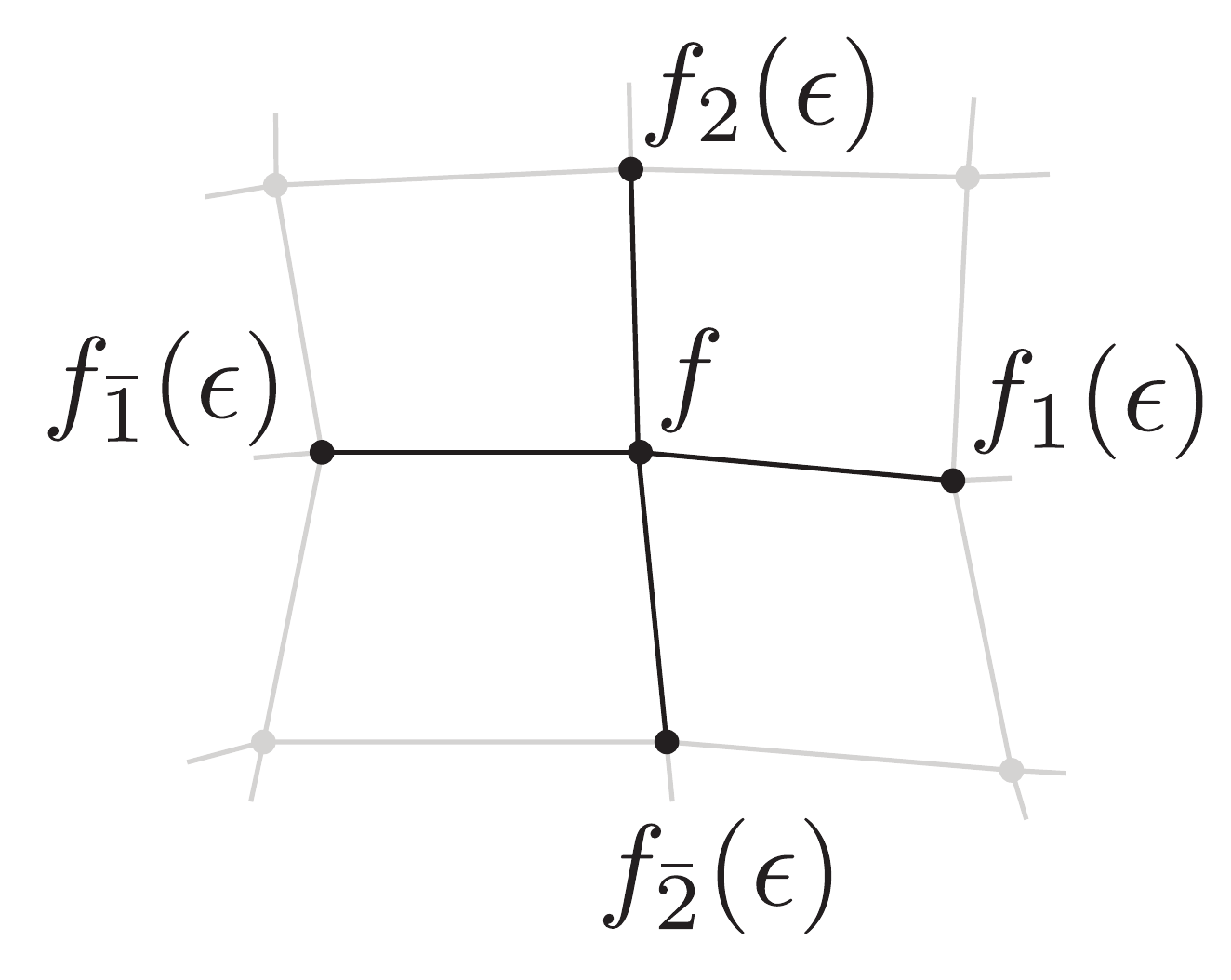}
\end{wrapfigure}
From here on, we refer to this set of points as an \emph{$\epsilon$-star} of the net $f$ around the point $p$ (see inset). The unit-length directions of the star edges are denoted as $\delta_jf(\epsilon),\ \delta_{\bar j}f(\epsilon)$.

By \defref{def:def1}, an $\epsilon$-star is a discrete orthogonal geodesic star if all its angles are equal, i.e., if
 \begin{equation}\label{eq:eps_star_angles}
 \langle \delta_j f(\epsilon),\ \delta_{j+1}f(\epsilon) \rangle  - \langle \delta_{j+1}f(\epsilon), \delta_{j+2}f(\epsilon) \rangle = 0, 
 \end{equation}	
where we use the notation $\delta_3f(\epsilon) = \delta_{\bar 1}f(\epsilon)$ and $\delta_4f(\epsilon) = \delta_{\bar 2}f(\epsilon)$ to enumerate all incident edges.
We show that our discretization is indeed loyal to the smooth case in the following theorem.

\begin{theorem}\label{Thm:taylor_orth_geo} \label{THM:TAYLOR_ORTH_GEO}  Equal angles on $\epsilon$-stars.
\begin{enumerate}
\item An analytic net $f$ is an orthogonal net, meaning $f_x \bot f_y$, if and only if all  its $\epsilon$-stars are discrete orthogonal geodesic stars up to first order, i.e., $\langle \delta_j f(\epsilon),\ \delta_{j+1}f(\epsilon) \rangle  - \langle \delta_{j+1}f(\epsilon),\  \delta_{j+2}f(\epsilon) \rangle = o(\epsilon).$
\item An analytic net $f$ is an orthogonal geodesic net if and only if all its $\epsilon$-stars are discrete orthogonal geodesic stars up to second order, i.e., $\langle \delta_jf(\epsilon),\ \delta_{j+1}f(\epsilon) \rangle  - \langle \delta_{j+1}f(\epsilon),\ \delta_{j+2}f(\epsilon) \rangle = o(\epsilon^2).$
\end{enumerate}
\end{theorem}
The proof is detailed in Appendix \ref{app:Taylor_proof}.

\subsection{Rigidity through developable surface extension}  \label{dev_extension}
Applying a deformation on a smooth developable surface locally generally dictates its shape globally. One way to see this is by looking at the rulings: 
 on a smooth developable surface, the rulings are global, in the sense that they either extend infinitely, or their endpoints must hit the boundaries of the surface \cite{spivak}. Flipping this point of view, one can ask how to extend a developable surface at its boundary: the possibilities are generally quite limited, since the points along the rulings are uniquely determined (see \figref{fig:rulings}). Note that arbitrarily extending rulings often results in singularities. Our discrete model shares a similar rigid structure, as shown in the following.

\subsubsection{Extension of a discrete orthogonal geodesic net}
Assume we have a vertex $F$ in our discrete net, as well as some neighboring vertices to its left (or right) and bottom (\figref{fig:reflection_and_cone}, right). The position of the top neighbor $F_2$ is then generally uniquely determined, as shown by the following two lemmas. Therefore, a given discrete orthogonal net can generally be extended at its boundary by setting only a small number of parameters, as illustrated in \figref{fig:evolution}. The process is analogue to the smooth case explained above, but it is not based on rulings.


\begin{lemma}\label{Lem:evo_plane_ref} (Direction propagation).
Given a vertex $F$ and three neighbors $F_{1},F_{\bar 2}, F_{\bar 1}$ such that the discrete curve $\Gamma_1$ through $F,F_{1},F_{\bar 1}$ is non-degenerate. Then there is a unique direction $\delta_{2}F$ such that $F,F_{1},F_{\bar 2},F_{\bar 1},F_{2}$ is an orthogonal geodesic star (where $F_2$ lies on the ray through $\delta_{2}F$; see \figref{fig:reflection_and_cone}).
\end{lemma}	
\begin{proof}
By \theoremref{Thm:orth_geo_equiv}, the vector $\delta_{2}F$ must be in the direction of the reflection of $\delta_{\bar 2}F$ w.r.t.\ the plane $\Pi_1$ spanned by $F,F_{1},F_{\bar 1}$.
\end{proof}
In the case where $\Gamma_1$ is a straight line, there is a family of solutions consisting of all vectors that are orthogonal to $\Gamma_1$.

\begin{figure}[b]
\centering
   \includegraphics[width=\linewidth]{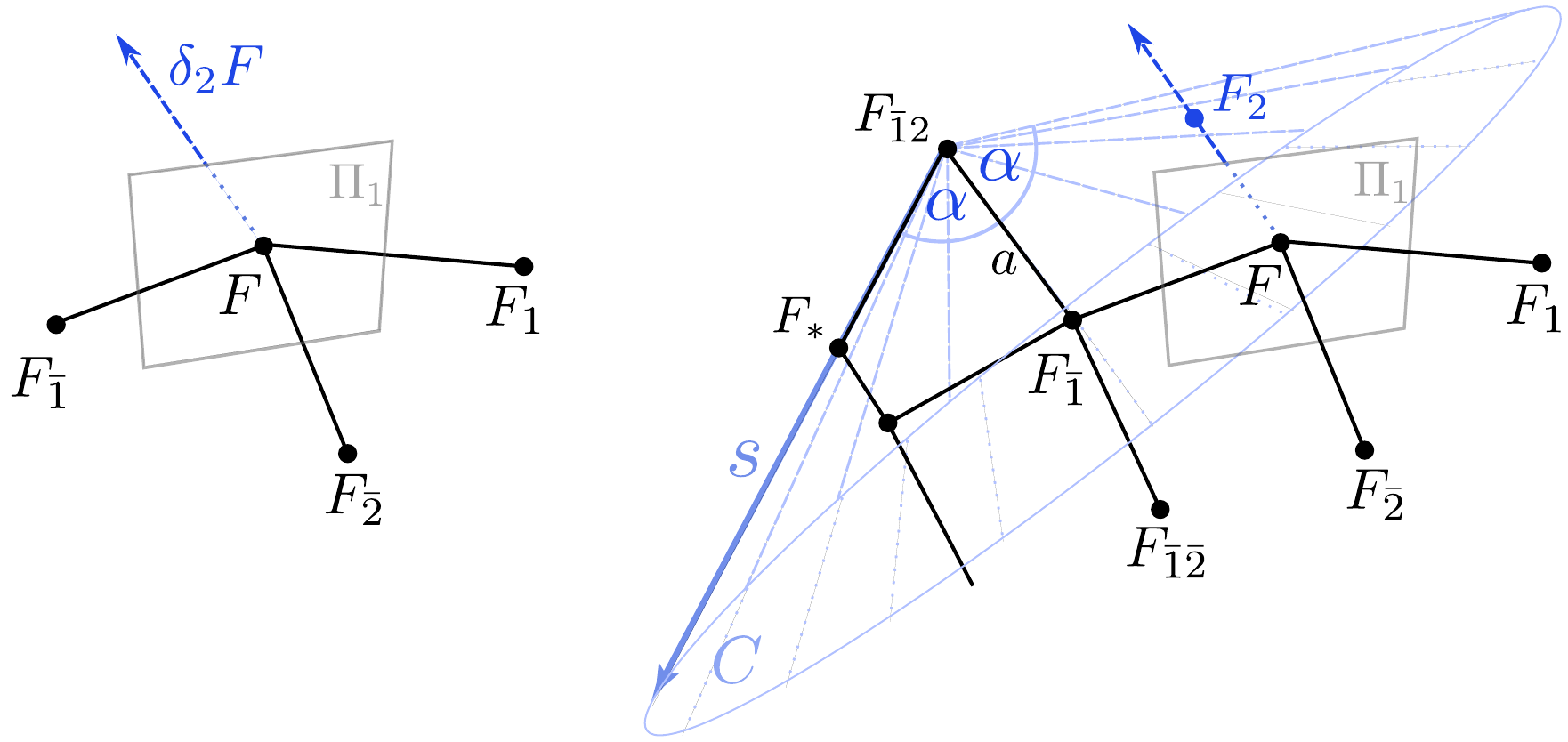}
   \caption{Left: By \lemmaref{Lem:evo_plane_ref} the direction $\delta_2 F$ is the reflection of $\delta_{\bar 2}F$ w.r.t.\ the plane $\Pi_1$. Right: By \lemmaref{Lem:evo_cone_int}, the same direction $\delta_2F$ intersects a cone $C$ with the apex at $F_{\bar1 2}$, determining the position of the point $F_2$. }
   \label{fig:reflection_and_cone}
\end{figure}

\begin{figure*}[t]
\centering
   \includegraphics[width=\linewidth]{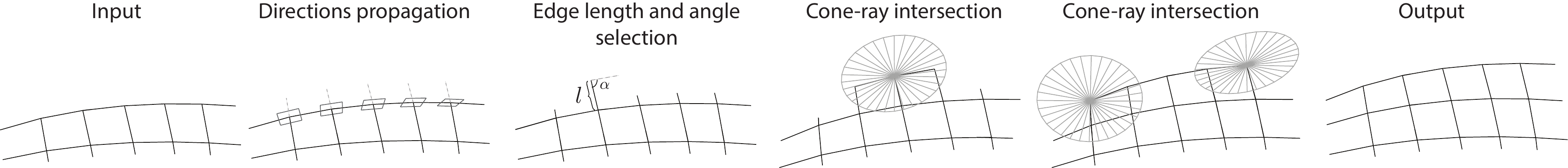}
   \caption{Extension of a discrete orthogonal net. Given a choice of two parameters: edge length $l$ and one angle $\alpha$, the row propagates by Lemmas \ref{Lem:evo_plane_ref} and \ref{Lem:evo_cone_int}.}
   \label{fig:evolution}
\end{figure*}

\begin{lemma}\label{Lem:evo_cone_int} (Cone-ray intersection).
Given a vertex $F$ in an orthogonal geodesic net that has at least all the neighboring vertices denoted in \figref{fig:reflection_and_cone} (right side). Let $C$ be the cone or plane generated by revolving the ray $s$ emanating from $F_{\bar 12}$ through $F_*$ about the axis $a = F_{\bar 1}- F_{\bar 12}$ (see \figref{fig:reflection_and_cone}). Then, the vertex $F_{2}$ has to lie on the intersection of $C$ and a line emanating from $F$ (\figref{fig:reflection_and_cone}).
\end{lemma}	
\begin{proof}
By \defref{def:def1}, the angle $\alpha$ between the net edges $a = F_{\bar 1}- F_{\bar 12}$ and $F_2 - F_{\bar12}$ must be equal to the angle between $a$ and $F_* - F_{\bar12}$, and so $F_{2}$ must lie on $C$.
\end{proof}

Given the construction for $F_2$ above, we see that, speaking informally,  extending a discrete orthogonal geodesic net by one vertex at its boundary is a determined process if we already have neighbors below and to the left or to the right. The only degrees of freedom are available when one begins adding a new row to the grid, without yet having neighbors on the left or right but only below, see \figref{fig:evolution}. Assuming general position, we first use \lemmaref{Lem:evo_plane_ref} to compute the directions of the new net edges  that point upwards. We can then select the length $l$ of the first new edge, effectively setting a vertex of the new row, as well as the cone half-angle $\alpha$ for the first cone $C$ of the new row. Then, the remaining vertices of the row are determined using \lemmaref{Lem:evo_cone_int}, as illustrated in \figref{fig:reflection_and_cone} and \figref{fig:evolution}.

\subsection{Relation to curvature line nets} \label{curve_geo_thm}
Here we prove a discrete version of the following simple theorem and connect discrete geodesic nets, conical nets and discrete orthogonal geodesic nets.
\begin{theorem} \label{Lem:geo_curv_smooth}
A smooth geodesic net $f$ that is also a curvature line net is an orthogonal geodesic net, and therefore a parameterization of a developable surface.
\end{theorem}
\begin{proof}
If $f$ is a curvature line net then $f_x$ and $f_y$ are orthogonal, hence by \corollaryref{cor:cont_orth_geo_dev} $f$ is developable.
\end{proof}

Conical meshes \cite{conical} are known to be a discrete analogue of curvature line nets. An inner vertex $v$ is conical if all the four oriented face planes meeting at $v$ are tangent to a common oriented cone of revolution, and a mesh is conical if its quads are planar and all of its inner vertices are conical.

\begin{theorem}\label{Lem:geo_curv_disc}
A discrete geodesic net $F$ that is also a conical net is a discrete orthogonal geodesic net.
\end{theorem}	
\begin{proof}
Using the notation of \figref{fig:geodesic_star}, a net is conical if and only if its quads are planar and every inner vertex satisfies the angle balance $\alpha_1 + \alpha_3 = \alpha_2 + \alpha_4$ \cite{conical_angle}. Since the net is also a discrete geodesic net, $\alpha_1 = \alpha_3$ and $\alpha_2 = \alpha_4$ and therefore $\alpha_1 = \alpha_2 = \alpha_3 = \alpha_4$, as in \defref{def:def1}.
\end{proof}

Note that both in the discrete and the smooth case, a (discrete) orthogonal geodesic net that is also a (discrete) conjugate net has planar coordinate curves.

%% file: 06-isometry-4Q.tex

\section{Discrete Isometry} \label{sec:iso}
 
\setlength{\intextsep}{8pt}%
\setlength{\columnsep}{8pt}%
\begin{wrapfigure}{r}{0.3\linewidth}
  \centering
  \includegraphics[width=\linewidth]{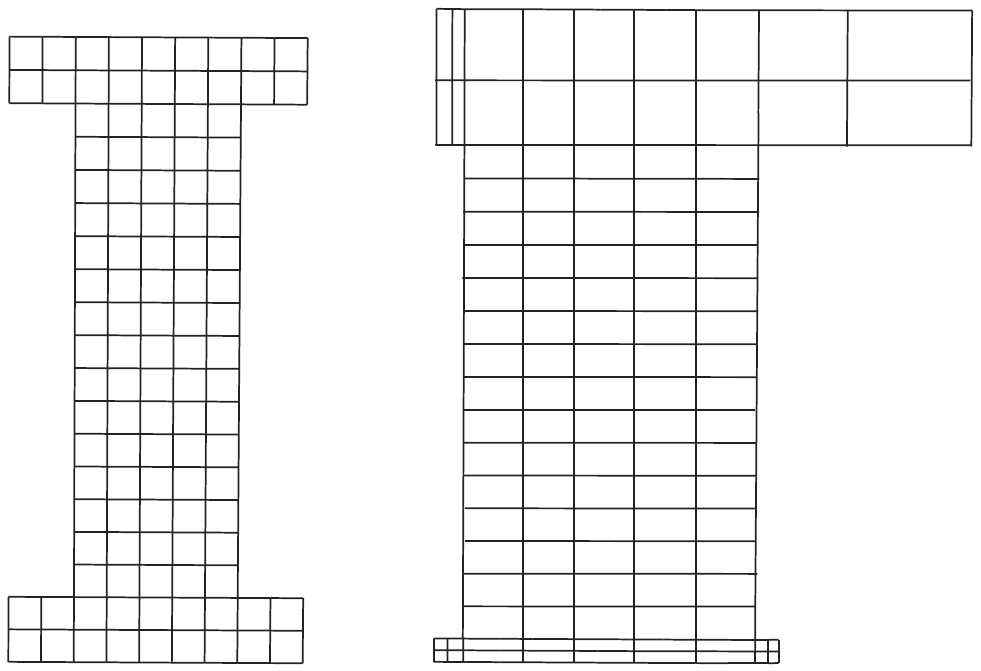}
\end{wrapfigure}
So far we have defined a model for discrete developable surfaces,
but we have not touched upon the subject of their \emph{discrete isometries}. 
Our net can describe a variety of surfaces with different scales, shapes and lengths (see inset for two orthogonal geodesic nets with the same connectivity). Though our editing system uses smoothness and isometry regularizers, which generally prevents large stretch in deformations, in this section we are looking for a definition of discrete isometry that specifies when two nets are ``the same'' in a precise manner.  Two smooth surfaces $S_1, S_2$ are said to be isometric, denoted $S_1 \cong S_2$, if there exists an isometry map $\phi: S_1 \rightarrow S_2$, i.e., a bijective map that preserves distances on the surfaces, or equivalently the lengths of all geodesics.

\begin{figure}[t]
\centering
   \includegraphics[width=0.7\linewidth]{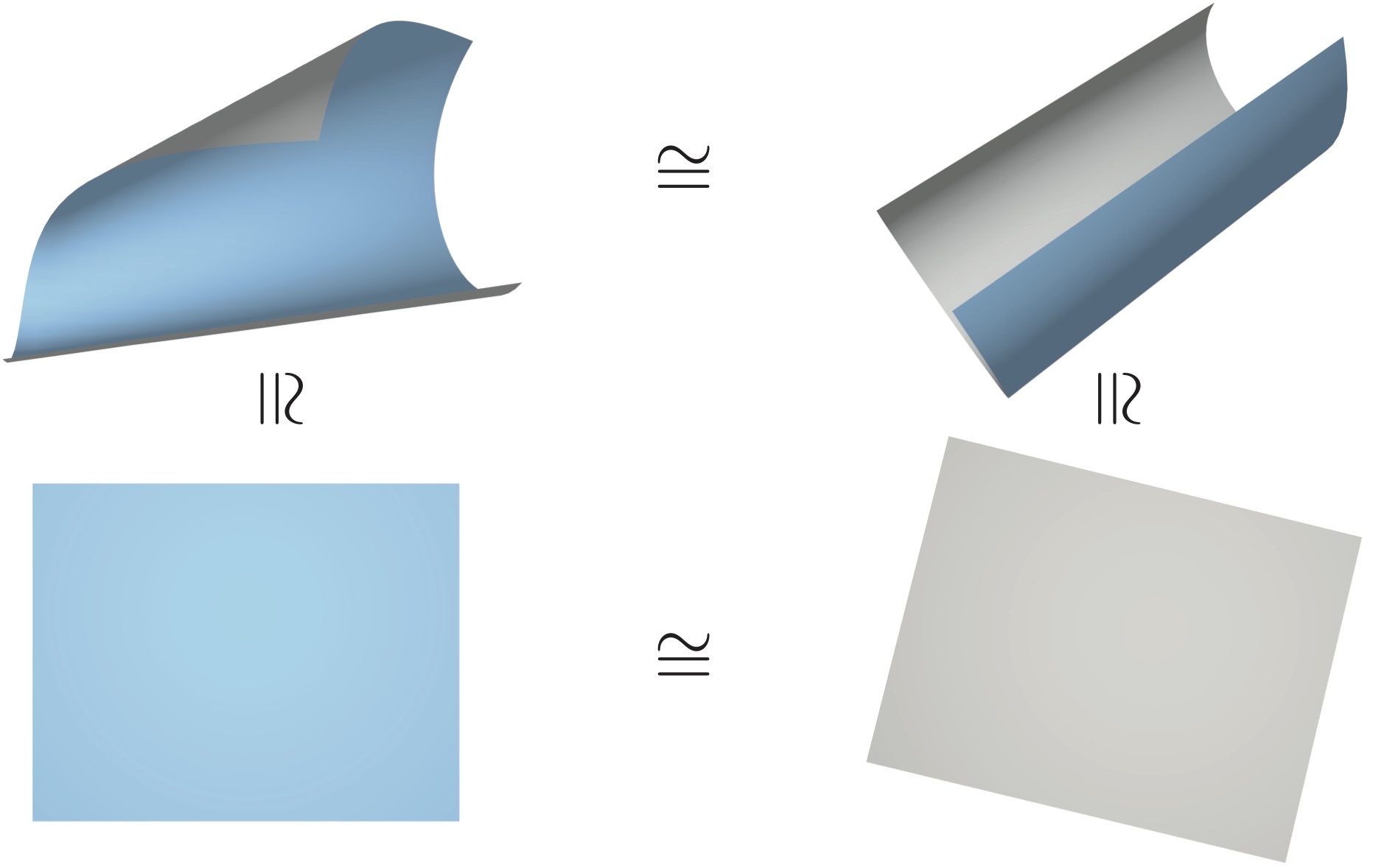}
      \caption{\label{fig:disc_iso}An application of \lemmaref{Lem:geo_disc_isometry}: A developable surface with disc topology and  piecewise geodesic boundary with $\frac{\pi}{2}$-corners is isometric to a flat rectangular shape on the plane. Two such surfaces with equal lengths of the boundary pieces  are isometric, as their flattened shapes are isometric.}
\end{figure}

\begin{figure*}[t]
\centering
   \includegraphics[width=\linewidth]{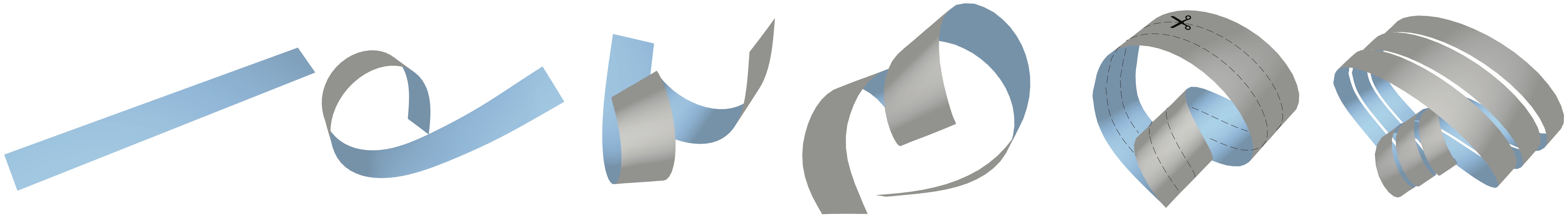}
   \caption{Our local isometry model allows us to deform, glue, and cut a surface while maintaining an exact discrete notion of isometry. In this figure we twist a long strip twice, glue it, and then cut it along two horizontal sections, creating three interleaved knotted surfaces.}
   \label{fig:iso_edit}
\end{figure*}

\subsection{Global isometry for disc topology nets}
In the special case of two developable surfaces with disc topology, one can test whether they are isometric by looking at their boundaries, as justified by the following lemma.
\begin{lemma} \label{Lem:geo_disc_isometry}  \label{LEM:GEO_DISC_ISOMETRY}
Let $S_1$ and $S_2$ be two  smooth developable surfaces with \emph{disc topology} and equal-length boundaries. Let $\gamma_1(s),\ \gamma_2(s)$ be their closed boundary curves in arc length parameterization and $\kappa_{g 1}(s),\ \kappa_{g 2}(s)$ the geodesic curvatures of these curves on $S_1$ and $S_2$, respectively. Then $S_1 \cong S_2    \iff \kappa_{g 1}(s)=\kappa_{g 2}(s)$.
\end{lemma}
\begin{proof}
See Appendix \ref{app:disc_iso_proof}.
\end{proof}
This lemma can be extended to the case of piecewise geodesic boundary, where the lengths of matching boundary pieces on the two surfaces are equal and the  angles of the turns (or ``corners'') match as well, see \figref{fig:disc_iso}. This is simple to discretize: two discrete developable nets $F_1$ and $F_2$ with disc topology and piecewise geodesic boundaries can be considered isometric if each matching pair of boundary pieces have equal lengths and the matching corners' angles agree. 

Such a global definition of isometry cannot be easily generalized to non-disc topologies and it does not provide us with the isometry map in the discrete case. One can easily find a situation where two
discrete nets $F_1, F_2$ with the same connectivity 
are deemed isometric by the global definition above, 
but there is no \emph{vertex-to-vertex} map $\Phi:F_1 \rightarrow F_2$ that we can reasonably call an isometry. 
\setlength{\intextsep}{8pt}%
\setlength{\columnsep}{8pt}%
\begin{wrapfigure}{r}{0.3\linewidth}
  \centering
  \includegraphics[width=\linewidth]{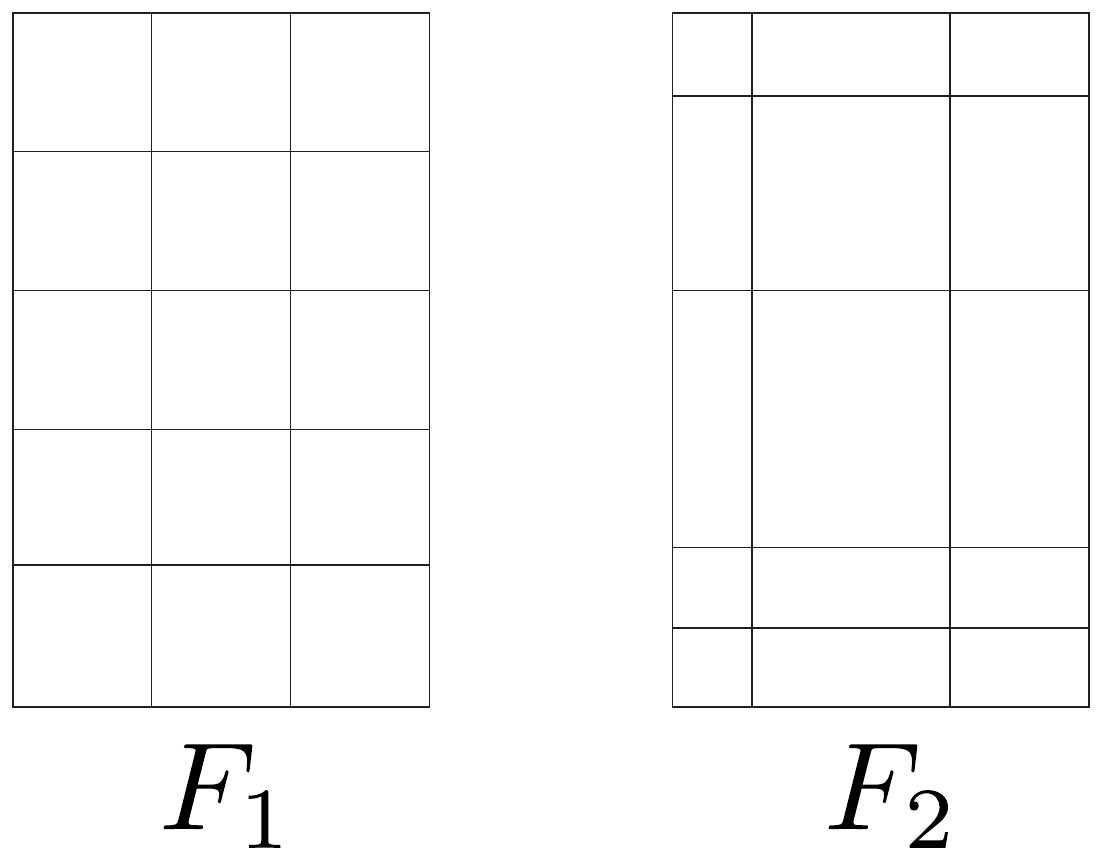}
\end{wrapfigure}
For example, the inset shows a case of two isometric rectangles represented by two different discrete orthogonal geodesic nets, where the discrete mapping $\Phi$ that matches corresponding vertices does not preserve any edge lengths. Consequently, a smaller piece $F_1' \subset F_1$ of the first surface is not isometric to the corresponding piece $\Phi(F_1') \subset F_2$ of the second surface. In practical terms, this means that the global criterion is too limited for the purposes of isometric shape modeling, and we need a \emph{local} definition of isometry that tells us when a mapping between two discrete nets is isometric.

\subsection{A local model for isometry: discrete orthogonal 4Q geodesic nets}
A natural attempt to define local isometry is to employ the global definition above to each local neighborhood on a surface. For our discrete nets, the first idea would be to look at the level of each single quad and impose length constraints. Unfortunately, the analysis in \secref{dev_extension} implies that we cannot add this many constraints to our net. \figref{fig:evolution} depicts how the cone-ray intersection discussed in \secref{dev_extension} propagates and determines a whole quad strip, leaving us solely one edge length and one angle per strip as degrees of freedom.
We therefore have to expand our notion of local neighborhood on discrete nets and loosen the developable net definition somewhat. We define a new class of nets called \emph{4Q orthogonal geodesic nets}, composed of \emph{4Q orthogonal patches}, defined as follows:

\begin{mydefinition} \label{def:4Q_patch}  An {orthogonal 4Q patch} is a composition of four quads (see \figref{fig:4Q}), such that:
\begin{tight_enumerate}
  \item Odd vertices have discrete orthogonal geodesic stars (\defref{def:def1}); \label{odd_orth}
  \item Even vertices have discrete geodesic stars (\defref{def:wunderlich}); \label{even_geo}
  \item The lengths of opposing sides (each a sum of two edges) of the 4Q patch are equal. \label{op_arcs_len}
\end{tight_enumerate}
\end{mydefinition}

\begin{figure}[h]
\centering
   \includegraphics[width=0.8\linewidth]{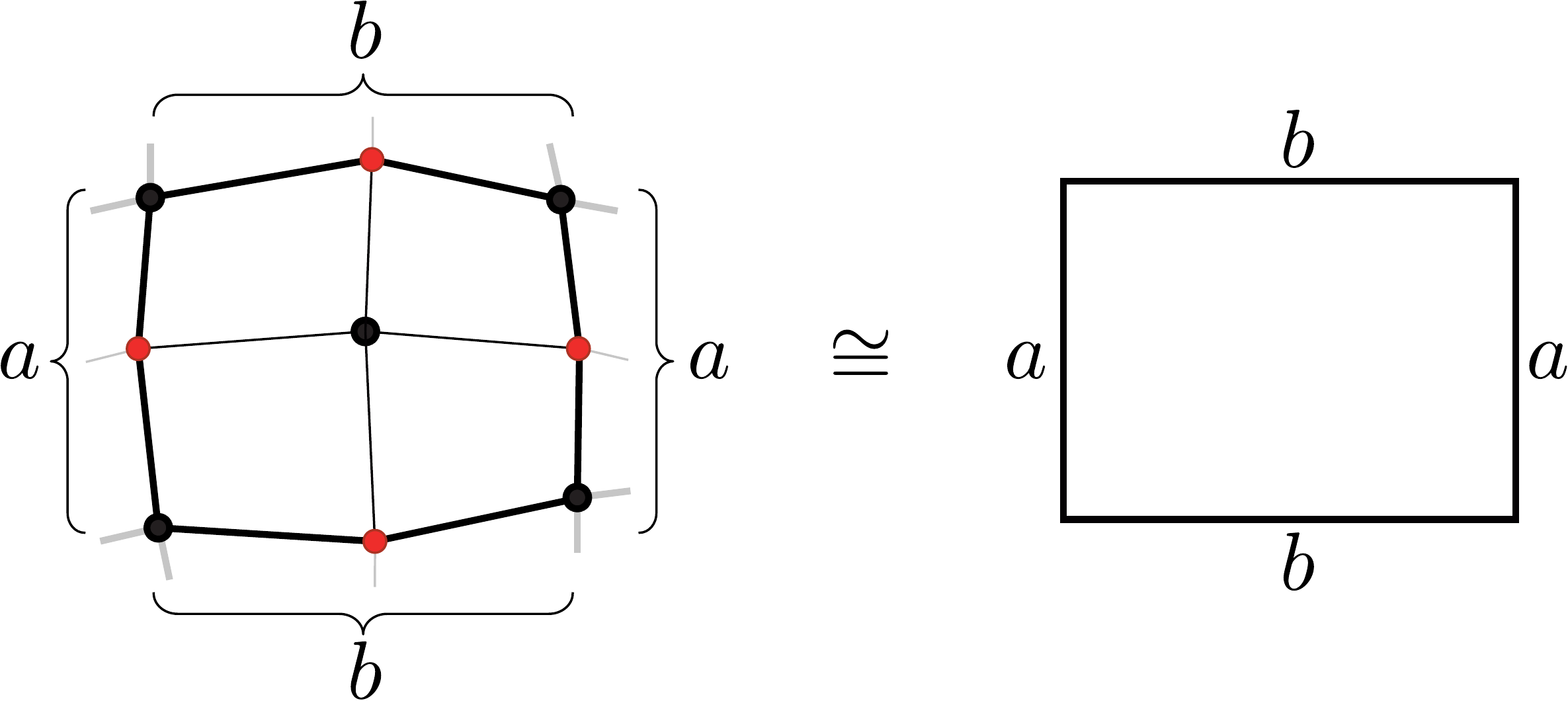}
      \caption{\label{fig:4Q}An orthogonal 4Q patch. Odd (black) vertices are discrete \emph{orthogonal} geodesic vertices, even (red) are discrete geodesic vertices. The lengths of opposing sides of the 4Q patch are equal. An orthogonal 4Q patch is seen as isometric to a rectangle in the plane with the same side lengths.}
\end{figure}

Conditions \eqref{odd_orth} and \eqref{even_geo} imply that an orthogonal 4Q patch can be seen as discrete developable, since its boundary can be interpreted as a set of four geodesic curves intersecting orthogonally, resulting in a vanishing integrated Gaussian curvature in the interior of the patch. 
Condition \eqref{op_arcs_len} implies that the 4Q patch can be seen as isometric to a rectangle, in the sense of the extension of \lemmaref{Lem:geo_disc_isometry} discussed above.
In the same spirit, we can model (global) isometries of the 4Q patch by requiring the conservation of the lengths of its sides.

An orthogonal 4Q geodesic net $F$ is a discrete net composed of orthogonal 4Q patches.
Two orthogonal 4Q geodesic nets are isometric if there exists a one-to-one correspondence between their 4Q patches, such that for each pair of matching  patches, the corresponding side lengths are equal. Modeling isometric deformations on an orthogonal 4Q net amounts to keeping these lengths fixed, enabling us to model isometries on a wide range of surfaces, unconstrained by their topology. 


In Appendix \ref{app:4Q_evo} we analyze the rigidity of orthogonal 4Q nets by looking at the construction of a 4Q net from a single strip, similarly to the analysis of orthogonal geodesic nets in \secref{dev_extension}. We observe that orthogonal 4Q nets have a similar rigid structure, which implies that while these nets do offer us additional degrees of freedom to incorporate local length constraints, they are not too permissive and still reasonably represent the space of developable surfaces.

\begin{figure*}[t]
\centering
   \includegraphics[width=\linewidth]{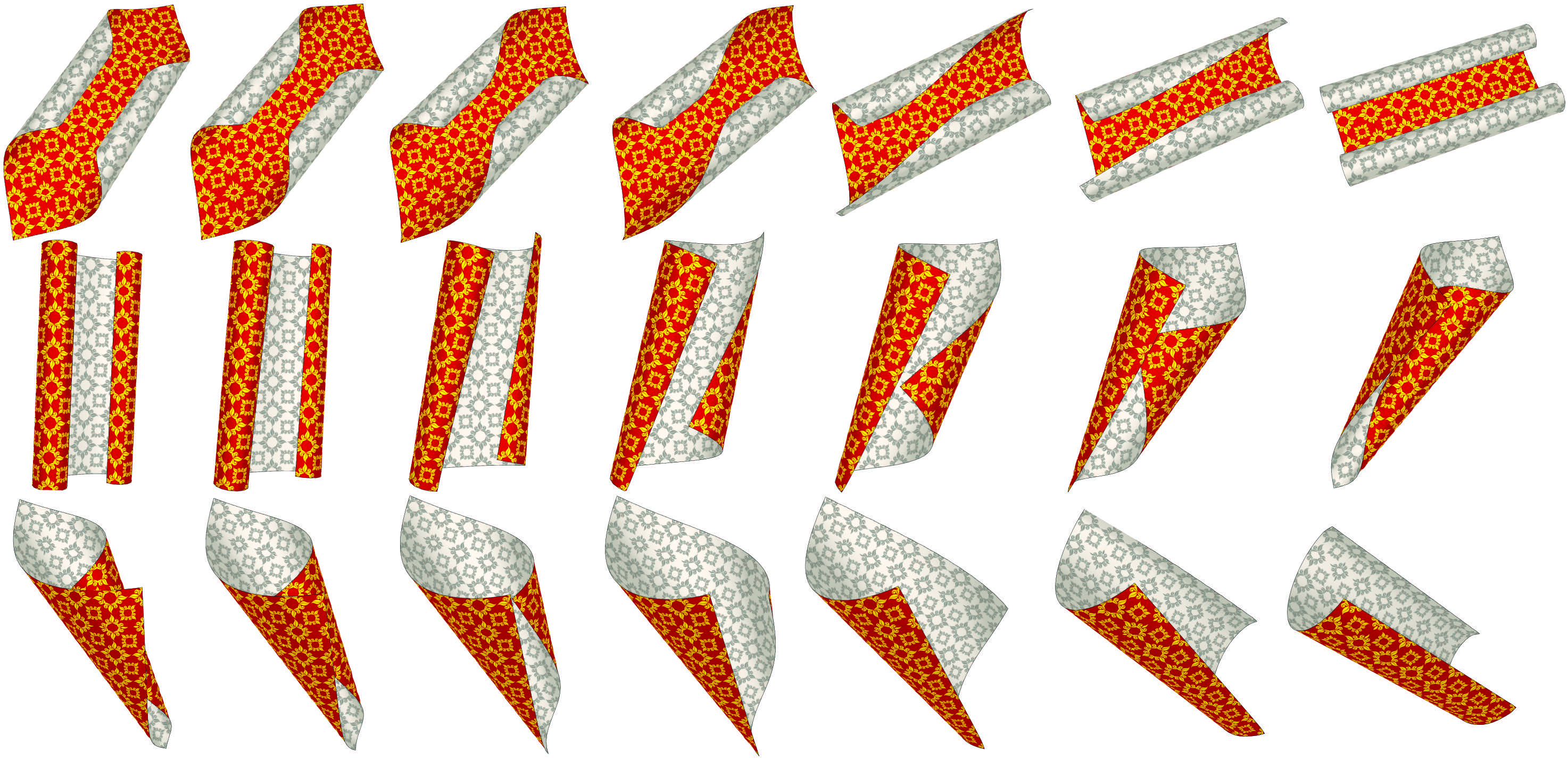}
   \caption{An interpolation sequence of isometric orthogonal 4Q nets. Note that in all cases, the deformations alter the rulings directions and their combinatorial structure. See the accompanying video for the entire sequence.}
   \label{fig:smooth-iso}
\end{figure*}

\subsection{Optimization} \label{sec:4q_opt}
To perform isometric surface deformation on orthogonal 4Q nets, our optimization stays largely similar to \secref{sec:opt}, with a few minor differences. We constrain the orthogonal geodesic vertices just as in \secref{opt_const} (\equref{eq:star_constraints}). Condition \eqref{even_geo} in \defref{def:4Q_patch}, i.e., equality of opposing angles around an even vertex can be written as
\begin{equation}
\label{eq:geo_constraints}
\langle e_j, e_{j+1} \rangle \|{e_{j+2}}\| \|{e_{j+3}}\| - \langle e_{j+2}, e_{j+3} \rangle \|{e_{j}}\| \|{e_{j+1}}\| = 0,
\end{equation}
where the $e_j$'s are the edge vectors emanating from the vertex. 
We combine the length constraints \eqref{op_arcs_len} in \defref{def:4Q_patch} with the isometry requirement by constraining the length of each side of each 4Q patch (i.e., the sum of the two respective edge lengths) to retain the same value as in the input orthogonal 4Q net.
We thus do not need to include an isometry regularizer as in \secref{sec:opt}, since our constraints already maintain the lengths of the coordinate curves exactly.

\subsection{Results}	
Incorporating the constraints in \secref{sec:4q_opt} allows us to isometrically edit orthogonal 4Q nets. We found experimentally that this optimization, which includes angle as well as length constraints, is in practice slower than the optimization in \secref{sec:opt}, allowing us to interactively edit coarser models of about 600 vertices. \figref{fig:iso_edit} demonstrates an editing operation that includes bending, glueing and cutting of a strip, all done while maintaining the orthogonal 4Q patches isometric to the reference state. 

Additionally, our constraints can be used in combination with a shape interpolation algorithm such as \cite{lipman2005,froh_botsch}. In \figref{fig:smooth-iso} we compute a sequence of isometric shapes, morphing a source shape into an (isometric) target, thereby simulating isometric bending of developable surfaces that generally happens \emph{not} along their rulings. An initial guess for each interpolation frame is first computed with \cite{froh_botsch}, followed by the optimization of (i.e., projection onto) our constraints, as specified in \secref{sec:4q_opt}. 

%% file: 07-conclusion.tex

\section{Limitations and future work}
This paper is a first step towards a discrete theory for modeling developable surface deformations through orthogonal geodesics. As such, this work focuses on the geometric model, its connections to the smooth case, and a straightforward integration of the model in existing applications. Various practical as well as theoretical problems remain unanswered, opening new avenues for further research, as detailed below.


\emph{Deformation algorithms for discrete developable geodesic nets.} Our most notable limitation is speed, as our editing system can only handle interactive editing of nets with ca.\ 1000 vertices. In this work we used an out-of-the box L-BFGS algorithm, and we leave it as future work to devise a more efficient deformation algorithm. In addition, we believe it would be useful to allow for interactive exploration of our shape space by discretizing various geometric flows, for instance to enable approximation of arbitrary shapes by our discrete developable nets.

\emph{Boundary conditions.} Our theory mainly concerns the internal vertices of the net, and our boundary constraints derived in \secref{opt_const} can be seen as a generalization of the internal vertex constraints, specifying that the boundary is a piecewise-geodesic curve, i.e., comprised of pieces of straight lines meeting at right angles. Currently, we can circumvent the jagged appearance of our boundaries by applying culling using alpha-textures, as was done for the letter G in \figref{fig:teaser} and is further illustrated in \figref{fig:flower_grid}. Given that developable surfaces are fairly rigid and the degrees of freedom in extending them at the boundary is quite limited, the culling approach is a reasonable pragmatic solution. Nevertheless, it would be interesting to derive other boundary conditions, allowing us to model curved boundaries with prescribed geodesic curvature using coarser models and represent shapes with curved boundaries by a tighter mesh.

\begin{figure}[t]
\centering
   \includegraphics[width=1\linewidth]{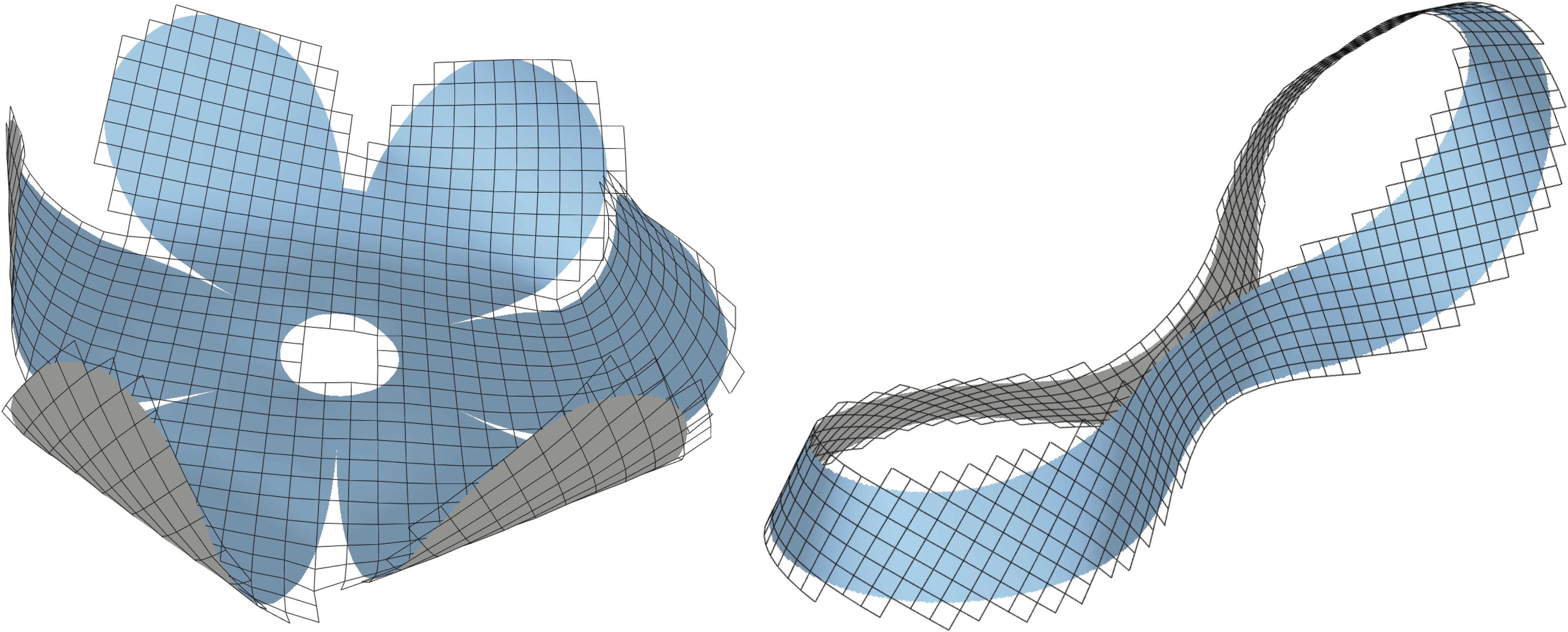}
   \caption{Editing a flower and an 'O' shaped developable surface with curved boundaries. Such boundaries can be approximated up to any precision by an orthogonal geodesic net. In practice however, for the purpose of interactive editing, our grid resolution is limited by our L-BFGS optimization. Our current pragmatic solution to alleviate the jagged boundary appearance is culling using alpha-textures. In the future we plan to explore the discretization of general curved boundary conditions and prescribing geodesic curvature.}
   \label{fig:flower_grid}
\end{figure}

\emph{Subdivision and refinement operations}. The geometry of our model consists solely of the vertex positions, and the quad faces are generally non-planar. Currently we simply arbitrarily triangulate the quad faces for rendering and fabrication purposes. In particular for fabrication applications, it would be interesting to look at refinement operations for our model that adhere to our constraints, as well as the convergence of such refinements to a smooth developable surface.

\emph{Discrete geodesic nets}. We leave further study of non-orthogonal discrete geodesic nets as future work. These can be beneficial for modeling developable surfaces, as well as deformations and isometries on more general doubly curved surfaces. In particular, we would like to define a discrete Gaussian curvature on these nets through an extension of the derivation in \secref{sec:disc_orth_nets}.

\emph{Isometry.} We are well aware that \secref{sec:iso} is just the tip of the iceberg. In terms of applications, modeling isometries is essential for simulating the bending of physical developable surfaces, and we have not yet experimented with methods to build or bend real life objects.
 We also did not treat the subject of choosing an optimal interpolation path between two isometric shapes, nor have we devised an interpolation algorithm with smoothness guarantees.
We believe that there is much more theory to explore in order to better understand the 4Q geodesic nets.

%% file: 96-appendix.tex
\section{Proof of \theoremref{Thm:taylor_orth_geo}} \label{app:Taylor_proof}
Assuming $f$ is analytic, with the shorthand $f_x = f_x(x,y), f_{xx} = f_{xx}(x,y)$, we use Taylor expansion to write the nearby points of $f$ in the form
\begin{align*}
  f_1(\epsilon) &= f + \epsilon f_x +\frac{\epsilon^2}2f_{xx} + o(\epsilon^3), & 
  f_{\bar1}(\epsilon) &= f - \epsilon f_x +\frac{\epsilon^2}2f_{xx} + o(\epsilon^3),\\
  f_2(\epsilon) &= f + \epsilon f_y +\frac{\epsilon^2}2f_{yy} + o(\epsilon^3), &
  f_{\bar2}(\epsilon) &= f - \epsilon f_y +\frac{\epsilon^2}2f_{yy} + o(\epsilon^3).
\end{align*}
The rest of the proof requires writing the first coefficients of the Taylor expansion of the edge directions $\delta_jf(\epsilon), \delta_{\bar j}f(\epsilon)$. Here we derive the coefficients of $\delta_1f(\epsilon)$, and the other coefficients are analogous.
The edge vector $f_1 - f$ can be written as
$$f_1(\epsilon) -f = \epsilon f_x +\frac{\epsilon^2}2f_{xx} + \dots$$
and so $\delta_1 f(\epsilon)$ can be written as
$$
\delta_1 f(\epsilon) = \frac{\epsilon  (f_x+\frac{\epsilon}2  f_{xx}+ \dots)}{\|{\epsilon  (f_x+\frac{\epsilon}2  f_{xx} + \dots)}\|} = \frac{f_x+\frac{\epsilon}2 f_{xx} + \dots}{\|{f_x+\frac{\epsilon}2  f_{xx} + \dots}\|}.
$$
Let $a_i$ be the Taylor coefficients of $\delta_1f(\epsilon)$, by direct computation:
\begin{align*}
& a_0 = \delta_1f(0) =  \frac{f_x}{\|{f_x}\|}\\
& a_1 = \delta_1f'(0) = -\frac{\langle f_{xx},f_x \rangle}{2 \langle f_x,f_x \rangle ^{3/2}}f_x+ \frac{1}{2\sqrt{\langle f_x,f_x \rangle}}f_{xx}.
\end{align*}
Similarly performing this for $\delta_{\bar 1}f(\epsilon), \delta_2 f(\epsilon), \delta_{\bar 2}f(\epsilon)$ and plugging the expressions in \equref{eq:eps_star_angles} gets us
$$
\langle \delta_1f(\epsilon), \delta_2 f(\epsilon) \rangle- \langle \delta_2 f(\epsilon), \delta_{\bar 1} f(\epsilon) \rangle = \frac {2\langle f_x,f_y\rangle}{\|{f_x}\| \|{f_y}\|} + o(\epsilon)
$$
and by symmetry we get exactly the same for the other angles. Therefore, the angles of an $\epsilon$-star are equal up to first order if and only if $f$ is an orthogonal (not necessarily geodesic) net. If $f$ is orthogonal, then by plugging in $\langle f_x,  f_y\rangle=0$ we see that:
\begin{align*}
&\langle \delta_1f(\epsilon), \delta_2 f(\epsilon) \rangle = \epsilon \frac{\langle f_x,f_{yy} \rangle + \langle f_{xx},f_y \rangle}{2\|{f_x}\| \|{f_y}\|} + o(\epsilon^2)\\
&\langle \delta_2f(\epsilon), \delta_{\bar 1}f(\epsilon) \rangle = \epsilon \frac{-\langle f_x,f_{yy} \rangle + \langle f_{xx},f_y \rangle}{2\|{f_x}\| \|{f_y}\|} + o(\epsilon^2)\\
&\langle \delta_{\bar 1}f(\epsilon), \delta_{\bar 2} f(\epsilon) \rangle = \epsilon \frac{-\langle f_x,f_{yy} \rangle - \langle f_{xx},f_y \rangle}{2\|{f_x}\| \|{f_y}\|} + o(\epsilon^2)\\
&\langle \delta_{\bar 2}f(\epsilon), \delta_{1} f(\epsilon) \rangle = \epsilon \frac{\langle f_x,f_{yy} \rangle - \langle f_{xx},f_y \rangle}{2\|{f_x}\| \|{f_y}\|} + o(\epsilon^2)
\end{align*}
Equality of all the linear terms implies $\langle f_x, f_{yy} \rangle = 0$ and $\langle f_y, f_{xx} \rangle = 0$. Together with $f_x \bot f_y$, this implies that $f$ is a geodesic orthogonal net. To see that, let $n^x$ be the principle normal of the $x$ coordinate curve and let  $f_{xx} = af_x + bn^x$ for some $a,b \in \R$. Then $0 = \langle f_{xx} , f_y \rangle = \langle af_x +bn^x, f_y \rangle = \langle bn^x, f_y \rangle$ and so $n^x \perp f_y$. By construction, the principle normal satisfies $n^x \perp f_x$, which means that the principle normal of the $x$ coordinate curve is parallel to the surface normal and so the curve is a geodesic. By a similar calculation, the principle normal of the $y$ coordinate curve is parallel to the surface normal.

\section{Proof of \protect{\lemmaref{Lem:geo_disc_isometry}}} \label{app:disc_iso_proof}
Every developable surface is \emph{locally} isometric to a planar surface. By \cite{pottmann_new}, a simply connected developable surface is (globally) isometric to a planar surface. Hence, disc topology developable surfaces $S_1,S_2$ are isometric to some planar surfaces $\hat{S}_1,\ \hat{S}_2$. As geodesic curvature is invariant to isometries, the curvatures of the boundary curves of $\hat{S}_1,\ \hat{S}_2$ are $\kappa_{g1}(s),\kappa_{g2}(s)$. By the fundamental theorem of planar curves, the planar boundary curves differ by a rigid motion (meaning that $\hat{S}_1, \hat{S}_2$ are exactly the same planar shape up to rigid motion) if and only if $\kappa_{g1}(s)=\kappa_{g2}(s)$, hence if and only if $S_1 \cong \hat{S}_1 \cong \hat{S}_2 \cong S_2$.

\begin{figure}[b]
\includegraphics[width=\linewidth]{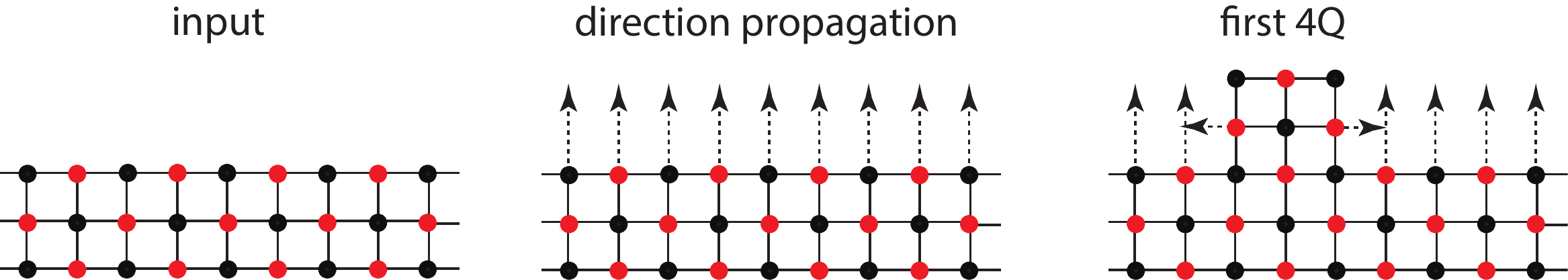}
\caption{\label{fig:4Q_evo_input}Left: A given 4Q strip. Center: By direction propagation, any vertex with its three neighbors can be generally completed to a geodesic star by a point on a unique ray. Right: The first extension 4Q quad must be such that the horizontal rays emanating from its middle, determined by the direction propagation, intersect the two neighboring vertical rays emanating from the strip, such that valid vertices can be formed at the intersection points.}
\end{figure}

\section{4Q net evolution} \label{app:4Q_evo}
Analogously to our analysis in \secref{dev_extension}, we show how orthogonal 4Q net constraints propagate from a given horizontal strip, leaving only a few degrees of freedom, and in practice, for nets representing smooth shapes, almost none. 
Recall that we denote by black vertices the centers of discrete orthogonal geodesic stars, while red vertices are centers of discrete geodesic stars that are not necessarily orthogonal; opposite sums of edges in every 4Q quad are equal. We start by noting that a vertex and three of its neighbors can be generally completed to a geodesic star by a point located on a unique ray (\figref{fig:4Q_evo_input}), and we refer to this as \emph{direction propagation}; this is analogous to the plane reflection \lemmaref{Lem:evo_plane_ref} that refers to the special case of orthogonal geodesic stars, and is a direct result of  \lemmaref{Lem:discrete_geo_n}. We analyze the most constrained case, where one 4Q quad is already given that extends our horizontal strip. This is similar to the choice of one edge length and angle for discrete orthogonal geodesic nets in \secref{dev_extension}, but with a few more degrees of freedom (\figref{fig:4Q_evo_input}). By direction propagation, this first extension 4Q quad  must be such that the two horizontal rays emanating from its middle intersect the two neighboring vertical rays from the strip (\figref{fig:4Q_evo_input}), so that valid vertices can be formed at the intersection points.

 We continue observing how the entire strip propagates by the orthogonal 4Q  geodesic net constraints. We refer the reader to \figref{fig:4Q_ray_cone}, where we note that by the previous constraint on the first extending 4Q quad, two rays intersect at a new vertex. The rest of the figure shows repeated application of \lemmaref{Lem:evo_cone_int}, a sequence of cone-ray intersections. Note that this lemma is also valid when only one of the vertices is a geodesic, as evident in its proof.
\begin{figure}
\includegraphics[width=\linewidth]{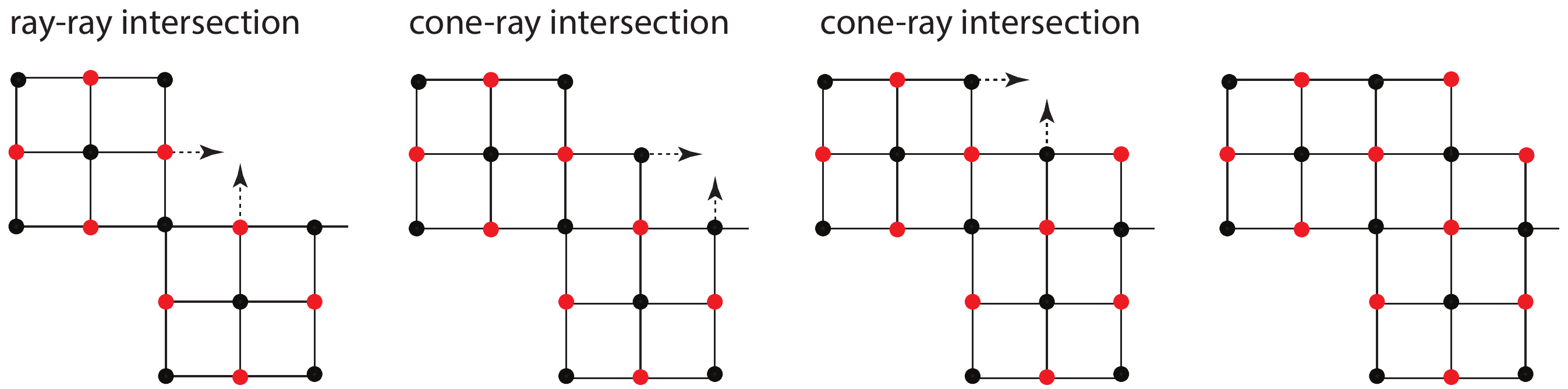}
\caption{\label{fig:4Q_ray_cone} A ray-ray intersection and repeated application of \lemmaref{Lem:evo_cone_int} determines the location of all but one vertices of a neighboring 4Q quad.}
\end{figure}

The cone intersection propagation determines all vertices of a neighboring 4Q quad but one. This vertex must fulfill two conditions: 
\begin{enumerate}
  \item The sums of edge lengths of the opposing vertical sides of the 4Q quad are equal;
  \item The sums of edge lengths of the opposing horizontal sides of the 4Q quad are equal.
\end{enumerate}
As all edges except one are already determined, this means that the missing vertex should lie in a fixed distance from two different points, or equivalently on an intersection of two spheres (\figref{fig:4Q_spheres}). If the spheres intersect, they either intersect in a point or a circle; in practice for a smooth enough net, this generally results in a circle. Every point on this circle satisfies the length constraint, but does not in general create a direction that intersects with a given vertical direction for the net. The set of all of these directions generates a cone, and so the last 4Q vertex lies on the intersection of this cone with a given vertical ray (see \figref{fig:4Q_spheres}). 
 This process repeats to reveal the entire extension strip, as the next vertex of a neighboring 4Q quad is given by a ray-ray intersection.

\begin{figure}
\includegraphics[width=\linewidth]{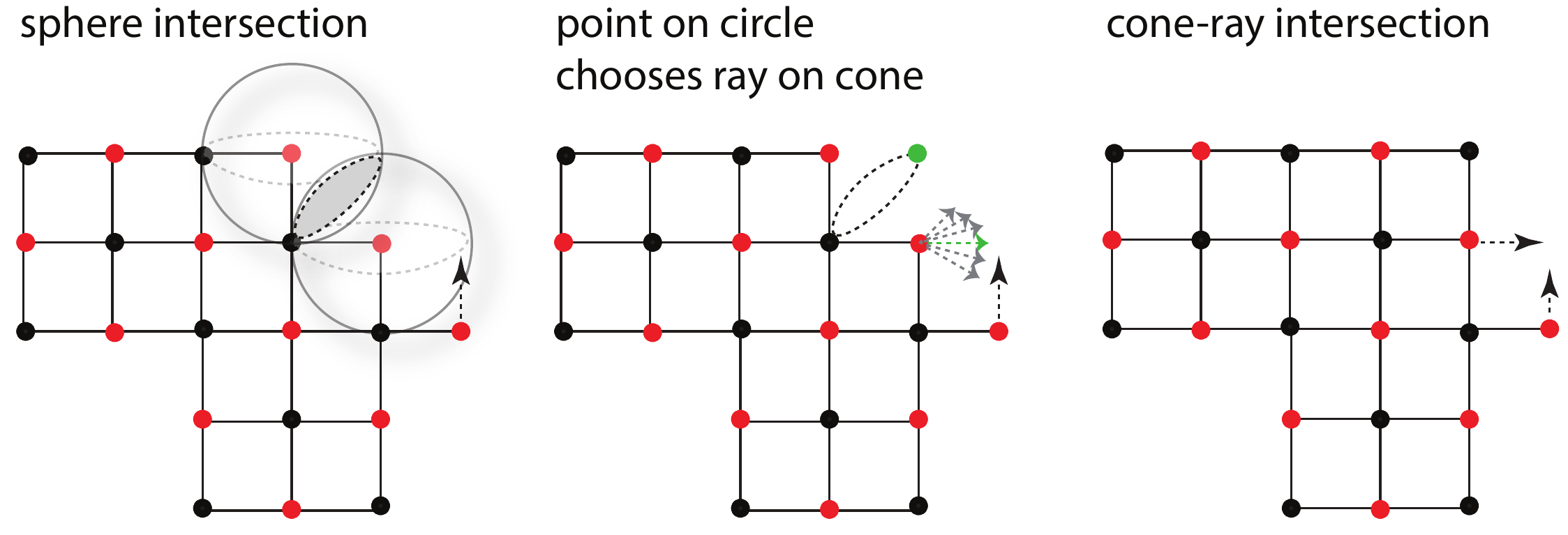}
\caption{\label{fig:4Q_spheres}
Left: The upper right corner vertex lies on an intersection of two spheres. Center: These spheres generally intersect in a circle. Every point on this circle determines a unique ray by direction propagation, and all these directions together form a cone. Right: This cone intersects with a given vertical ray, and the upper right corner vertex is a point on a circle that propagates the direction of the intersecting ray on the cone.
}	
\end{figure}